\documentclass[11pt]{article}

\usepackage[utf8]{inputenc}             %
\usepackage[english]{babel}             %
\usepackage{csquotes}
\usepackage[T1]{fontenc}                %
\usepackage{authblk}                    %
\usepackage[margin=1.2in]{geometry}
\usepackage{amsmath}
\usepackage{amssymb}                    %
\usepackage{amsfonts}                   %
\usepackage{amsthm}                     %
\usepackage{amscd}                      %
\usepackage{algorithm}
\usepackage{algpseudocode}
\usepackage{nicefrac}
\usepackage{enumitem}

\usepackage[hidelinks]{hyperref}
\usepackage{mathtools} 
\mathtoolsset{showonlyrefs=true}        %

\usepackage{graphicx}                   %
\usepackage{tikz}                       %
\usetikzlibrary{calc}
\usepackage{multirow}                   %
\usepackage{xcolor} %
\usepackage{colortbl} %
\usepackage{subcaption} %

\usepackage[linesnumbered, algoruled, vlined, algo2e]{algorithm2e}
\SetKwInput{KwInput}{Input}
\SetKwInput{KwOutput}{Output}

\usepackage[
    backend=biber,
    natbib=true,
    style=numeric-comp,
    sorting=none,
    maxnames=2,
    giveninits=true,
    url=false,
    arxiv=pdf,
    doi=false,
    eprint=true,
    isbn=false,
    date=terse,
    block=ragged,
    maxbibnames=99
]{biblatex}
\numberwithin{equation}{section}
\theoremstyle{plain}

\newtheorem{theorem}{Theorem}[section]
\newtheorem{lemma}[theorem]{Lemma}

\newtheorem{remark}[theorem]{Remark}

\newtheorem{example}[theorem]{Example}

\newcounter{assumptcounter}
\newenvironment{assumpt}{%
    \refstepcounter{assumptcounter}
    \begin{enumerate}[align=left, leftmargin=3.5em]
        \item[\textbf{Assumption~\theassumptcounter.}]
    }{\end{enumerate}}

\newcommand{\eps}{\varepsilon}  %
\newcommand{\Eps}{{\scalebox{0.8}{\ensuremath{\mathcal{E}}}}}
\newcommand{\Eta}{{\scalebox{1.2}{\ensuremath{\eta}}}}

\newcommand{\muref}{\mu_{\mathrm{ref}}}
\newcommand{\mutar}{\mu_{\mathrm{tar}}}

\newcommand{\piref}{\pi_{\mathrm{ref}}}
\newcommand{\pitar}{\pi_{\mathrm{tar}}}

\newcommand{\Dkl}{\ensuremath{D_{\mathrm{KL}}}}

\newcommand*{\dd}{\ensuremath{\,\mathrm{d}}}  %
\newcommand*{\dx}[1]{\ensuremath{\dd{#1}}}  %

\DeclareMathOperator{\grad}{\nabla}
\newcommand{\normal}[2]{\ensuremath{\mathcal{N}\big({#1},{#2}\big)}}  %

\newcommand{\given}{\,\vert\,}
\newcommand{\smallbullet}{%
    \raisebox{-0.8ex}{\scalebox{2.1}{$\cdot$}}%
}

\colorlet{LightGray}{black!15!white}
\colorlet{Gray}{black!40!white}
\definecolor{AmericanRose}{rgb}{1.0, 0.01, 0.24}
\definecolor{ForestGreen}{rgb}{0.13, 0.55, 0.13}
\colorlet{DarkGreen}{ForestGreen!78!black}
\definecolor{MidnightBlue}{rgb}{0.1, 0.1, 0.44}
\definecolor{Orchid}{rgb}{0.85, 0.44, 0.84}
\definecolor{Pumpkin}{rgb}{1.0, 0.46, 0.09}
\definecolor{Beaver}{rgb}{0.62, 0.51, 0.44}
\definecolor{Teal}{rgb}{0.0, 0.5, 0.5}
\definecolor{NavyBlue}{rgb}{0.02, 0.5., 0.8}
\colorlet{BlueGray}{Gray!40!NavyBlue}

\definecolor{PTBBlue}{rgb}{0.01953125, 0.6171875, 0.83984375}
\definecolor{PythiaGreen}{rgb}{0.5625, 0.65234375, 0.0625}
\definecolor{QumphyViolet}{rgb}{0.63921569, 0.48627451, 0.96470588}
\definecolor{QumphyOrange}{rgb}{0.85098039, 0.65882353, 0.30980392}

\newcommand{\todo}[1]{}
\newcommand{\nando}[1]{}
\newcommand{\philipp}[1]{}
\newcommand{\maren}[1]{}
\newcommand{\sebastian}[1]{}

\newcommand{\old}[1]{}

\title{Estimating systematic errors in Bayesian inversion using transport maps}

\author[1]{Maren Casfor}
\author[1]{Philipp Trunschke}
\author[1]{Nando Hegemann}
\author[1,2]{Sebastian Heidenreich}

\affil[1]{Physikalisch-Technische Bundesanstalt, Abbestr. 2--12, 10587 Berlin, Germany}
\affil[2]{corresponding author: sebastian.heidenreich@ptb.de}
\date{\today}

\begin{document}

\maketitle

\begin{abstract}
  In indirect measurements, the sought parameters have to be determined by solving an inverse problem, typically in a Bayesian framework.
  Often, the accurate numerical simulation of the measuring process is computationally demanding, making it necessary to rely on approximate models.
  These surrogates, however, introduce an additional model error and thus may distort the resulting parameter distribution.
  Moreover, even with the additional speed granted by the surrogate, posterior determination through conventional means such as Markov chain Monte Carlo might be cost intensive, specifically for complicated posterior shapes.
  In this paper, we propose a unified framework that combines Bayesian inference, model error correction and a transport-based sampling scheme to address these issues. 
  To train the transport scheme, we investigate two different losses: one equivalent to the Kullback-Leibler divergence associated to the transport problem and one based on an upper bound of this loss, generally known as the evidence lower bound.
  We demonstrate that training the transport based on the latter changes the optimisation landscape drastically, potentially introducing an undesired bias in approximating the target posterior.
  We compare the computational cost of our approach with established methods and underline the theoretical results with numerical examples.
\end{abstract}

\section{Introduction}\label{sec:introduction}
From healthcare to engineering, indirect measurements play an important role in our modern society.
They allow to obtain information about quantities that are not directly accessible.
Prime examples include medical imaging or process control in semiconductor fabrication. 
To evaluate indirect measurements, a model is necessary to reconstruct the sought parameters by solving an inverse problem.
These inverse problems are typically ill-posed, and the collected data are intrinsically noisy due to the involved devices, their setup and the measurement process itself~\cite{Tarantola2005_Inverseproblemtheory}.
To ensure traceability of measurements, i.e.\ a complete set of associated uncertainties for each step of the measurement process~\cite{JCGM2021_Internationalvocabularymetrology}, it is necessary to determine both the unknown parameters and their uncertainties.
The statistical method of Bayesian inversion has proven to be very efficient in solving inverse problems and providing reconstructions of the sought parameters as well as their uncertainties~\cite{Stuart2010_InverseproblemsBayesian}.

There exists a variety of numerical methods to solve Bayesian inverse problems.
Among these, \emph{Markov chain Monte Carlo} (MCMC) methods~\cite{Brooks2011_HandbookMarkovchain} are particularly popular and frequently used because they are widely applicable and, unlike grid-based discretisations, do not suffer from an exponential growth in computational cost with dimension.
Nevertheless, in practical applications the efficiency of MCMC algorithms can reduce with increasing dimension~\cite{Robert1999_MonteCarlostatistical}.
Moreover, MCMC methods are prohibitively slow, even for low dimensions, if evaluating the model is computationally expensive.
To overcome this issue, alternative approaches such as variational inference or sequential Monte Carlo have been proposed to find approximations for the posterior distribution~\cite{Blei2017_Variationalinferencereview,DelMoral2006_Sequentialmontecarlo}.

In an alternative approach the posterior distribution is characterised by a transport map, which is a push-forward map from a simple reference distribution to the posterior~\cite{el2012bayesian,parno2018transport}.
The advantage of this approach is that, once the transport map has been determined, further samples from the posterior distribution can be obtained very efficiently.
Furthermore, if the evaluation of multiple measurements with posteriors similar in shape is required, the transport approach is very efficient.
This might for example be the case in production control of mechanical parts or semiconductors.

Another aspect of inverse problems is the reliability of the simulated measurement process.
Measurements are typically modelled by a deterministic \emph{forward model} and stochastic noise, the \emph{observational noise}.
However, it may be that not all information is known and the model is erroneous~\cite{Scholz2025_Quantificationmodelerror}.
It may also be that the exact forward model is too complex for practical applications, introducing the need for an efficient approximation.
In both cases the resulting \emph{model error} causes a bias in the posterior distribution and thus in the parameter estimation~\cite{Kaipio2007_Statisticalinverseproblems,Brynjarsdottir2014_Learningphysicalparameters,Plock2023_Impactstudynumerical}.
The interplay between model and measurement errors and their effect on biases has been investigated in detail for Bayesian inversion in~\cite{Kennedy2001_Bayesiancalibrationcomputer}.

In the present paper, we introduce an approach that unifies Bayesian inversion, transport maps and the iterative model error approach from~\cite{Calvetti2018_Iterativeupdatingmodel} to tackle inverse problems involving computationally expensive forward models afflicted by modelling errors.
Our proposed approach iteratively minimises a \emph{Kullback--Leibler} (KL) divergence involving non-linearly nested expectations.
In machine learning, this challenge is often addressed by optimising an easier to compute upper or lower bound of the target objective.
However, it is not clear if the resulting loss function exhibits the same minimiser as the original objective.
In this paper, we investigate an upper bound loss (\emph{Jensen--loss}) resulting from the application of Jensen’s inequality as well as a \emph{nested Monte Carlo loss}, which is equivalent to the original objective. 
Herein, we also demonstrate that the \emph{Jensen--loss} can introduce a significant bias.
We demonstrate the effectiveness of our approach on two simple yet illustrative examples: an affine forward model and the model of a simple machine, first introduced in~\cite{Brynjarsdottir2014_Learningphysicalparameters}.
We show that a bias caused by the model error is corrected by the proposed algorithm, but that deviations from the correct variance remain.

The paper is structured as follows.
Section~\ref{sec:methods} recalls fundamental concepts and introduces the model error along some main properties.
Based on that, we develop an abstract version of the proposed iterative sampling scheme.
Section~\ref{sec:algorithm} introduces a practical way to perform this iteration by means of transport maps.
Additionally, we compare the computational costs of the proposed approach and compare it to MCMC-based alternatives.
In Section~\ref{sec:num_experiemnts}, we apply the method for two examples and finally conclude with a discussion of our findings.

\section{Fundamental concepts}\label{sec:methods}

In this section, we present three fundamental concepts we use in the course of this article. 
The first topic is the concept of Bayesian Inversion and the basic assumptions we use here. 
In the second part, we introduce the used model error approach from~\cite{Calvetti2018_Iterativeupdatingmodel} and contextualise it with other existing approaches. 
To give a justification for this approach, we investigate some basic properties of the model error distribution from which we derive the assumptions necessary for our proposed method.
We then demonstrate how these assumptions naturally lead to the proposed iterative method.
In the last paragraph, we give a short introduction into the topic of transport maps, which we use in order to sample from the posterior distribution.

\subsection{Bayesian inversion}%
\label{subsec:methods:bayesian_inversion}

Many practical applications require solving \emph{inverse problems}, i.e.\ to find an unknown parameter $x$ to a mathematical model $F$ given observations of the model's output
\begin{equation}
    \label{eq:methods:bayesian_inversion:inverse_problem_noiseless}
    y = F(x)~.
\end{equation}
The parameter $x \in \mathcal{X}$ and the \emph{observation} $y \in \mathcal{Y}$ are assumed to be elements of separable Banach spaces $\mathcal{X}$ and $\mathcal{Y}$, typically Euclidean spaces of different but finite dimensions. 
The observation operator, called \emph{forward model}, $F \colon \mathcal{X} \to \mathcal{Y}$ is assumed to be measurable with respect to the induced Borel algebras.

Inverse problems are often ill-posed~\cite{Stuart2010_InverseproblemsBayesian}, meaning that a solution may not exist, may not be unique, or that small changes in the observation $y$ may lead to extreme, even discontinuous, changes in the sought parameter $x$.
This sensitivity to the observation $y$ is particularly problematic since $y$ is often afflicted by an \emph{observational noise} $\eta$ and given by
\begin{equation}
    \label{eq:methods:bayesian_inversion:inverse_problem_exact_model}
    y = F(x) + \eta~.
\end{equation}
However, in many applications $\eta$ is a measurement noise which can naturally be modelled as a random variable with known distribution $\Eta\sim\mu_\Eta$.
In this setting, we can mitigate the ill-posedness of the inverse problem by modelling $x$ as a random variable $X$, encoding our prior beliefs about $x$ in its \emph{prior distribution} $\mu_X$, implicitly regularising the problem~\cite{Stuart2010_InverseproblemsBayesian,Kaipio2005_StatisticalComputationalInverse,Dashti2015_BayesianApproachInverse}.
Then,~\eqref{eq:methods:bayesian_inversion:inverse_problem_exact_model}, which must now be considered as the realizations of the corresponding random variables, defines the $X$- and $\Eta$\hbox{\,-\,}dependent random variable $Y$ with distribution $\mu_{Y|X,\Eta} = \delta_{Y = F(X) + \Eta}$. Thus, Bayes’ formula can be used to calculate the \emph{posterior distribution}
\begin{equation}
\label{eq:methods:bayesian_inversion:bayes_posterior}
    \dd\mu_{X|Y}(x \vert y)
    = \frac{1}{\pi_{Y}(y)} \pi_{Y|X}(y \vert x) \dd \mu_X(x)
    ~,
\end{equation}
with the \emph{likelihood function} $\pi_{Y|X}(y \vert x) = \pi_{\Eta}(y-F(x))$ and the observation-dependent normalisation constant $\pi_{Y}(y) = \int_{\mathcal{X}} \pi_{Y|X}(y|x) \dd \mu_X(x)$.
For the sake of simplicity, we assume here that all measures are absolutely continuous with respect to the Lebesgue measure $\lambda$ and denote their corresponding densities by $\pi = \frac{\dd\mu}{\dd\lambda}$, with appropriate subscripts.\footnote{The reader is referred to~\cite[Section~7.4]{Klenke2020} for an introduction to probability density functions.}

Although choosing $\mu_X$ and $\mu_{\Eta}$ is non-trivial (especially in the absence of prior knowledge about $x$) and can impact the posterior distribution $\mu_{X|Y}$ significantly (cf.~\cite{Gustafson1996_Localsensitivityposterior,Heidenreich2015_statisticalinverseproblem}), the present work considers these choices to be fixed.
In addition to that, we assume for the observational noise to be independent of $X$, as it is a reasonable assumption for many applications \cite{Farchmin2019_efficientapproachglobal,Andrle2021_anisotropyopticalconstants}.
We summarise these beliefs in the subsequent two assumptions.

\begin{assumpt}
\label{assumpt:methods:bayesian_inversion:prior_assumption}
    The prior distribution $\mu_X$ is fixed and encodes all knowledge about $X$.
\end{assumpt}

\begin{assumpt}
\label{assumpt:methods:bayesian_inversion:noise_assumption}
    The observational noise $\Eta$ is independent of $X$ and the distribution $\mu_\Eta$ is fixed.
\end{assumpt}

\subsection{Model error}%
\label{subsec:methods:model_error}

For given prior and noise distribution the posterior~\eqref{eq:methods:bayesian_inversion:bayes_posterior} can be evaluated pointwise, but calculating quantities of interest, like marginals, is infeasible.
In order to have access to the posterior we need to use sample-based approaches such as MCMC.
These algorithms make use of many posterior evaluations, which leads to many function evaluations of the forward model~$F$.
In many applications, however, the exact forward model $F \colon \mathcal{X} \to \mathcal{Y}$ is unknown or computationally expensive.

To circumvent this problem, a reduced forward model $f \colon \mathcal{X} \to \mathcal{Y}$, which is easier to evaluate, can be used and the \emph{reduced inverse problem} 
\begin{equation}
    \label{eq:methods:model_error:reduced_inverse_problem}
    y = f(x) + \eta
\end{equation}
is then solved instead of~\eqref{eq:methods:bayesian_inversion:inverse_problem_exact_model}.
This simplification introduces an error that propagates to the posterior distribution, as illustrated in~\cite{Brynjarsdottir2014_Learningphysicalparameters}.
The resulting error can not be captured by the observational noise $\Eta$, which covers only the statistical error of the measurement process.
Ignoring the parameter-dependent systematic error can, hence, lead to a significant overestimation of $\eta$ and conflicts with the interpretation of $\eta$ as an $X$-independent measurement noise.
This section investigates an additional error term originating from a potential systematic error arising due to inaccuracies in the forward model.
Several strategies exist to compensate for this error \cite{Kennedy2001_Bayesiancalibrationcomputer,Calvetti2018_Iterativeupdatingmodel,Jia2018_InfinitedimensionalBayesian}, which can mainly be formalised by introducing a new random variable, the \mbox{\emph{model error} $\Eps$}, to the model, yielding
\begin{equation}
\label{eq:methods:model_error:inverse_problem_with_model_error}
    y = f(x) + \varepsilon + \eta~.
\end{equation}
We encode this statistical model in the subsequent assumption.
\begin{assumpt}
\label{assumpt:methods:model_error:distribution_of_observation_given_by_delta_distribution}
    The observation is distributed according to $\mu_{Y|X,\Eps,\Eta} = \delta_{Y=f(X) + \Eps + \Eta}$.
\end{assumpt}
Using the definition of the model error function
\begin{equation}
    \label{eq:methods:model_error:def_model_error_function_M}
     M(x) := F(x) - f(x)~
\end{equation}
as in \cite{Calvetti2018_Iterativeupdatingmodel}, we can see that the model~\eqref{eq:methods:model_error:inverse_problem_with_model_error} interpolates between the exact inverse problem~\eqref{eq:methods:bayesian_inversion:inverse_problem_exact_model}, with $\Eps = M(X)$ corresponding to $\mu_{\Eps | X} = \delta_{\Eps = M(X)}$,
and the reduced inverse problem~\eqref{eq:methods:model_error:reduced_inverse_problem}, with $\Eps = 0$ corresponding to $\mu_{\Eps} = \delta_{\Eps=0}$.
Another approach, commonly used in application, is the choice $\Eps\sim \mu_{\Eps} = \normal{0}{\Sigma_{\Eps}}$, where the covariance matrix $\Sigma_{\Eps}$ is estimated as well (cf.~\cite{Heidenreich2015_BAYESIANAPPROACHSTATISTICAL,Henn2012}).
This approach is easy to apply, but especially the assumption, that $\Eps$ is a Gaussian variable centred in $0$ is a strong restriction.
To motivate a more general approach presented in~\cite{Calvetti2018_Iterativeupdatingmodel} we note that 
choosing $\Eps = M(X)$ is intuitively the most accurate way to describe the model error, but since we end up with the exact inverse problem~\eqref{eq:methods:bayesian_inversion:inverse_problem_exact_model}, it requires many evaluations of the exact forward model $F$ (cf.~Table~\ref{fig:algortithm:compare_cost:table_costs}).
Our choice for $\mu_{\Eps}$ should therefore require less evaluations of the exact forward model and moreover contain some reasonable properties describing the model error.
In this context we note that the joint distribution $\mu_{Y,X,\Eps,\Eta}$ uniquely determines the Bayesian posterior $\mu_{X|Y}$ and hence should be used to define $\mu_{\Eps}$.

Subsequently, we discuss two intuitive properties of the exact model error distribution $\mu_{\Eps}$ for $\Eps = M(X)$ which we like to impose on our simplified version of $\mu_{\Eps}$.
Both properties follow directly from the definition of the model error $\Eps = M(X)$.
Since $\Eps$ depends in a deterministic way on the parameter $X$, given knowledge of $X$, additional knowledge of the observed data $Y$ should not influence our prediction of $\Eps$.
Moreover, as the model error depends only on $X$, its posterior distribution should be completely determined by the posterior distribution of $X$, formally
$$
    \mu_{\Eps | Y}(A \,|\, Y)
    = \mathbb{P}[\Eps\in A \,|\, Y]
    = \mathbb{P}[X\in M^{-1}(A) \,|\, Y]
    = \mu_{X | Y}(M^{-1}(A) \,|\, Y)
    = M_\# \mu_{X | Y}(A \,|\, Y) .
$$
The described properties can be summarized as two assumptions for the joint distribution listed below.

\begin{assumpt}
    \label{assumpt:methods:model_error:eps_givenX_indep_of_noise_meas}
    The systematic error $\Eps$ is a function of the input $X$, independent of the observation $Y$, i.e.~$\mu_{\Eps|Y,X} = \mu_{\Eps|X}$.
\end{assumpt}

\begin{assumpt}
    \label{assumpt:methods:model_error:eps_givenY_defined_by_M_posterior}
    The posterior systematic error $\Eps|Y$ is the image of the posterior input $X|Y$ under $M$, i.e.~$\mu_{\Eps|Y} = M_\#\mu_{X|Y} \coloneqq \mu_{X|Y} \circ M^{-1}$.
\end{assumpt}
Since $\Eps = M(X)$ models reality most faithfully, we argue that it is natural to impose these assumptions onto our simplified model error distribution.
To provide rigorous support for the validity of these assumptions, we consider the joint distribution $\mu^{\star}_{Y,X,\Eps,\Eta}$ corresponding to $\Eps = M(X)$.
The subsequent lemma proves that $\mu^{\star}_{Y,X,\Eps,\Eta}$ satisfies Assumptions~\ref{assumpt:methods:model_error:eps_givenX_indep_of_noise_meas} and~\ref{assumpt:methods:model_error:eps_givenY_defined_by_M_posterior}.

\begin{table}[ht]
    \centering
    \begin{subtable}[t]{0.54\textwidth}
    \resizebox{\textwidth}{!}{
        \begin{tabular}{|p{2.4cm}|p{5.2cm}|p{3.2cm}|}
            \hline
            ~ & ~ & \\
            ~               &     \textbf{less knowledge required} & \textbf{more knowledge required} \\
            \hline
            & & \\[-5pt]
            \textbf{independent of $X$} &
                \begin{tabular}{l}
                     $\smallbullet ~\Eps = 0$ \\
                     $\smallbullet ~ \Eps \sim \mathcal{N}\left(0\, ,\, \sigma_{\varepsilon}^2I_{d_{\mathcal{Y}}} \right)$,\\
                     \hspace{0.2cm} for $\sigma_{\Eps}\in \mathbb{R}$
                \end{tabular}
                &
                \begin{tabular}{c}
                        \tikz[overlay,remember picture] \node (cell22) {};\\
                     \textcolor{QumphyOrange}{$\smallbullet~\boldsymbol{\Eps \sim M_{\#} \pi_{X \given Y}}$}\\
                     \hspace{0.2cm}\textcolor{QumphyOrange}{\textit{\small\textbf{(our model)}}}
                \end{tabular}\\[1cm]
            \hline
            & & \\[-5pt]
            \textbf{dependent of $X$} &
                \begin{tabular}{l}
                     $\smallbullet ~ \varepsilon \sim \mathcal{N}\left(0\, ,\, \Sigma_{\Eps}  \right)$,  \\
                     \hspace{0.2cm}{\small for $\Sigma_{\Eps} = \operatorname{diag}\Big( ( a f(X) ) ^2 + b^2 \Big)$,}\\
                     \hspace{0.2cm}{\small $a,b \in \mathbb{R}$}
                \end{tabular}
                &
                \begin{tabular}{l}
                     $\smallbullet ~ \Eps \sim M(X)$\\[0.4cm]
                     $\smallbullet ~ \Eps \sim GP(0 \, , \, K) $
                \end{tabular}\\[1cm]
        \hline 
        \multicolumn{1}{c}{\begin{tabular}{c}
             \\[1.6cm]
        \end{tabular}}
        \end{tabular}
        \begin{tikzpicture}[overlay,remember picture]
            \draw[BlueGray, thick,line width=1.1pt, minimum height=4em, minimum width=3em] ([shift={(-4em, 0.2em)}]cell22) rectangle ++(8.2em, -7.4em);
            \node at ({cell22 |- 0em,+6.2em})  {\textcolor{BlueGray}{ \textit{ \small \bfseries $\boldsymbol{F}$ known } }};
        \end{tikzpicture}
        } %
    \caption{model error matrix}
    \label{tab:methods:model_error:table_models_overview}
    \end{subtable}
    \hfill
    \begin{subtable}[t]{0.45\textwidth}
    \resizebox{\textwidth}{!}{
        \begin{tabular}{|p{4.4cm} | p{4.5cm}|}
        \hline
        \multicolumn{2}{|c|}{
        \begin{tabular}{c}
        ~ \\[-5pt]
        \textbf{Bayesian modelling assumptions} \textit{(given)}\\
        $\pi_{Y \given X, \Eps ,\Eta}, \pi_X, \pi_{\Eta}$\\[0.2cm]
        \end{tabular}
        } \\
        \hline
            \begin{tabular}{c}
                ~\\
                exact model for $\Eps$\\
                $\mu_{\Eps \given X} = \delta_{\Eps = M(X)}$\\
                $ \Downarrow $\\
                {
                \setlength{\arrayrulewidth}{1.2pt}
                \begin{tabular}{|c|}\arrayrulecolor{QumphyOrange}
                    \hline
                    \\[-5pt]
                    structural properties\\
                    \hspace{-0.4cm}$\smallbullet \pi_{\Eps \given Y , X} = \pi_{\Eps \given X}$\\
                    $\smallbullet \pi_{\Eps \given Y} = M_{\#} \pi_{X \given Y}$\\[5pt]
                    \hline
                \end{tabular}
                }\\
                ~ \\
            \end{tabular}
            & 
            {
            \setlength{\arrayrulewidth}{1.2pt}
            \begin{tabular}{|c|}
            \arrayrulecolor{QumphyOrange}
            \hline
            \\[-5pt]
            "black-box"-assumption\\
            $\pi_{\Eps \given X} = \pi_{\Eps}$\\[0.5cm]
            \hline
            \end{tabular} 
            }
            \\
        \hline
        \multicolumn{2}{c}{
        \begin{tabular}{c}
         \tikz[overlay,remember picture] \node (node_arrow) {};\\
         {
         \setlength{\arrayrulewidth}{1.2pt}
         \begin{tabular}{|c|}\arrayrulecolor{QumphyOrange}
            \hline
            \\[-5pt]
            ~~resulting model~~\\
            ~~$\pi_{\Eps} = M_{\#} \pi_{X \given Y }$~~\\[0.2cm]
            \hline
        \end{tabular}
        }
        \end{tabular}
        }
        \end{tabular}
        \begin{tikzpicture}[overlay, remember picture]
            \draw[->, QumphyOrange, thick ,line width=1.2pt] ([shift={(-8em, 2em)}]node_arrow.south) to  [out=280, in=160] ++(3.5em, -4em);
            \draw[->, QumphyOrange, thick, line width=1.2pt] ([shift={(8.3em, 4em)}]node_arrow.south) to  [out=-100, in=20] ++(-3.5em, -5.5em);
        \end{tikzpicture}
        } %
        \caption{modelling assumptions}
        \label{tab:methods:model_error:modelling_assumptions}
    \end{subtable}
    \caption{ The table~\ref{tab:methods:model_error:table_models_overview} provides an overview of some existing model error approaches,
    which are categorised according to whether they depend on $X$ and how much knowledge is required for the respective approach.
    Here we distinguish between the knowledge required to adequately model $\mu_{\Eps}$ (e.g. the selection of a hyperparameter) and the approach used here, where the exact forward model has to be known.
    In~\ref{tab:methods:model_error:modelling_assumptions} all assumptions used to define the model error distribution $\mu_{\Eps}$ are summarized.
    Hereby we divide the assumptions into general modelling assumptions for Bayes, referring to the prior and noise distribution and the conditional distribution for the observation (see Assumptions~\ref{assumpt:methods:bayesian_inversion:prior_assumption}, \ref{assumpt:methods:bayesian_inversion:noise_assumption}, \ref{assumpt:methods:model_error:distribution_of_observation_given_by_delta_distribution}), structured properties resulting from the exact model (see Assumption \ref{assumpt:methods:model_error:eps_givenX_indep_of_noise_meas}, \ref{assumpt:methods:model_error:eps_givenY_defined_by_M_posterior}) and the "black-box" assumption, that $\Eps$ is independent of $X$ (Assumption \ref{assumpt:methods:model_error:epsilon_independent_of_X}). The sum of all these assumptions leads to the model choice for the model error.
        }
\end{table}

\begin{lemma}
\label{lem:methods:model_error:properties_exact_model}
    It holds that $\pi^{\mathrm{\star}}_{\Eps| Y,X} = \pi^{\mathrm{\star}}_{\Eps| X}$ and $\pi^{\mathrm{\star}}_{\Eps| Y}(\Eps\given y) = M_\# \pi^{\mathrm{\star}}_{X|Y}(\Eps\given y)$.
\end{lemma}
\begin{proof}
    First, recall that $\mu^{\mathrm{\star}}_{\Eps|X} = \delta_{\Eps = M(X)}$ implies that $\Eps$ is independent of $\Eta$ and that $\Eps$ and $Y$ are independent, given $X$.
    To see this, note that
    \begin{align}
        \pi^{\mathrm{\star}}_{\Eps,Y| X}(\Eps,y\given x)
        &= \frac{\pi^{\mathrm{\star}}_{\Eps,Y,X}(\varepsilon,y,x)}{\pi^{\mathrm{\star}}_{X}(x)} \\
        &= \frac{\pi^{\mathrm{\star}}_{\Eps,X}(\varepsilon,x)}{\pi^{\mathrm{\star}}_{X}(x)}
        \frac{\pi^{\mathrm{\star}}_{\Eps,Y,X}(\varepsilon,y,x)}{\pi^{\mathrm{\star}}_{\Eps,X}(\varepsilon,x)} \\
        &= \pi^{\mathrm{\star}}_{\Eps| X}(\varepsilon\given x)
        \pi^{\mathrm{\star}}_{Y| \Eps,X}(y\given \varepsilon,x) \\
        &= \pi^{\mathrm{\star}}_{\Eps| X}(\varepsilon\given x)
        \pi^{\mathrm{\star}}_{Y| X}(y\given x)~,
    \end{align}
    where the last equality follows since $\Eps = M(X)$.\footnote{Since $\Eps = M(X)$, the generated $\sigma$-algebras satisfy $\sigma(\Eps, X) = \sigma(X)$.
    The conditional probability given $X$ and $\Eps$ thus satisfies $\mathbb{P}(A\given X,\Eps) = \mathbb{E}[1_{A}|X,\Eps] = \mathbb{E}[1_A|X] = \mathbb{P}(A\given X)$ for any measurable set $A\in\sigma(Y)$, implying $\mu_{Y|X,\Eps} = \mu_{Y|X}$.}
    
    With this in hand, it is easy to see that $\Eps$ is independent of $Y$, given $X$, i.e.
    $$
        \pi^{\mathrm{\star}}_{\Eps|Y,X}(\varepsilon\given y,x)
        = \frac{\pi^{\mathrm{\star}}_{\Eps,Y|X}(\varepsilon, y\given x)}{\pi^{\mathrm{\star}}_{Y|X}(y\given x)}
        = \pi^{\mathrm{\star}}_{\Eps|X}(\varepsilon\given x)~.
    $$
    Finally, this implies
    \begin{align}
        \pi^{\mathrm{\star}}_{\Eps|Y}(\varepsilon\given y)
        &= \int_{\mathcal{X}} \pi^{\mathrm{\star}}_{\Eps|Y,X}(\varepsilon\given y, x) \pi^{\mathrm{\star}}_{X|Y}(x\given y) \dd x \\
        &= \int_{\mathcal{X}} \pi^{\mathrm{\star}}_{\Eps|X}(\varepsilon\given x) \pi^{\mathrm{\star}}_{X|Y}(x\given y) \dd x \\
        &= \int_{\mathcal{X}} \delta_{\Eps=M(x)} \pi^{\mathrm{\star}}_{X|Y}(x\given y) \dd x \\
        &= M_\# \pi^{\mathrm{\star}}_{X|Y}(\varepsilon\given y)
    \end{align}
    and concludes the proof.
\end{proof}
The final simplifying assumption for our choice of $\mu_{\Eps}$ is the independence of $\Eps$ from $X$.
To motivate this, we note that according to the Doob--Dynkin lemma, any model choice where $\Eps$ does depend on $X$ implies $\Eps = h(X)$, for some (potentially random) function $h$.\footnote{Note that a deterministic $h$ can be represented by the most informative prior possible, a Dirac measure.}
Due to its high computational complexity, we have to disregard the optimal error model $h=M$.
We thus have to find an efficiently computable error model describing $\Eps$ adequately.\footnote{We appreciate the irony that doing so introduces just a different form of systematic error.}
For any choice of $h$, additional information about the properties of $\Eps$ is needed.
This can be in the form of a prior for $h$, such as Gaussian processes $h\sim \mathcal{GP}(0,K)$~\cite{Kennedy2001_Bayesiancalibrationcomputer,Brynjarsdottir2014_Learningphysicalparameters}, or in the form of a model class characterizing the behaviour of $h$.
In either way, the accuracy of the prediction strongly depends on these choices (cf.~\cite{Bai2024_GaussianprocessesBayesian} for $h\sim\mathcal{GP}(0,K)$), which introduces yet another potential source of systematic errors.
Here we note that choosing a parametric model for $h$ increases the dimension of the parameter space, which in turn deteriorates the convergence of the posterior~\cite{Vaart2004_Vitessesdeconvergence}.
Hence, using a parametric model for $h$ requires more information in form of more measurements, to compensate for the larger parameter dimension.
Since precise information about the model error are rarely known, and increasing the number of measurements can be very expensive, we decide to impose the following assumption on our model.

\begin{assumpt}
    \label{assumpt:methods:model_error:epsilon_independent_of_X}
    The systematic error $\Eps\sim\mu_{\Eps}$ is independent of $X$, i.e.\ $\mu_{\eps|X} = \mu_\eps$.
\end{assumpt}
This assumption interpolates between the most accurate model $\mu_{\Eps|X} = \delta_{\Eps = M(X)}$, which is not independent of $X$, and the simplest choice $\mu_{\Eps} = \delta_{\Eps=0}$, which is independent of $X$.
This independence leads to a black-box setting, where we do not have to make any adaptations to the model error for different applications.
An overview of existing models together with required information and assumptions is given in Table~\ref{tab:methods:model_error:table_models_overview}.
As visualized in Table~\ref{tab:methods:model_error:modelling_assumptions}, we assume for our model choice the two structural properties of the joint distribution $\mu_{Y,X,\Eps,\Eta}^{\star}$ given by Assumptions~\ref{assumpt:methods:model_error:eps_givenX_indep_of_noise_meas} and~\ref{assumpt:methods:model_error:eps_givenY_defined_by_M_posterior}, and the ``black-box''--assumption in Assumption~\ref{assumpt:methods:model_error:epsilon_independent_of_X}.
Then the model error distribution can be characterised in the following way. %

\begin{theorem}
\label{thm:methods:model_error:model_error_distribution_follow_from_assumptions}
    Assumptions~\ref{assumpt:methods:model_error:eps_givenX_indep_of_noise_meas},~\ref{assumpt:methods:model_error:eps_givenY_defined_by_M_posterior} and~\ref{assumpt:methods:model_error:epsilon_independent_of_X} imply
    \begin{equation}
    \label{eq:methods:model_error:model_error_distribution_as_pushforward_under_M}
        \pi_{\Eps}(\varepsilon) = M_\# \pi_{X|Y}(\varepsilon\given y) .
    \end{equation}
\end{theorem}
\begin{proof}
    Recall that Assumptions~\ref{assumpt:methods:model_error:eps_givenX_indep_of_noise_meas} and~\ref{assumpt:methods:model_error:epsilon_independent_of_X} require $\mu_{\Eps|Y,X} = \mu_{\Eps|X} = \mu_\Eps$ and, hence,
    \begin{align}
        \pi_{\Eps|Y}(\varepsilon\given y)
        &= \int_{\mathcal{X}} \pi_{\Eps|Y,X}(\varepsilon\given y, x) \pi_{X|Y}(x\given y) \dd x \\
        &= \pi_{\Eps}(\varepsilon) \int_{\mathcal{X}} \pi_{X|Y}(x\given y) \dd x \\
        &= \pi_{\Eps}(\varepsilon)~.
    \end{align}
    Assumption~\ref{assumpt:methods:model_error:eps_givenY_defined_by_M_posterior} ($\mu_{\Eps|Y} = M_\# \mu_{X|Y}$) thus implies $\pi_{\Eps}(\varepsilon) = M_\# \pi_{X|Y}(\varepsilon\given y)$.
\end{proof}

\begin{remark}
\label{rmk:mehtods:model_error:true_prediction_for_M_linear}
    One should keep in mind, that assuming $\Eps\sim\mu_{\Eps}$ to be independent of $X$ generally leads to a wrong prediction, since
    $$
        \mathbb{E}[Y | X] = f(X) + \mathbb{E}[\Eps] \ne F(X)
    $$
    for generic $\mu_{\Eps}$. %
    However, for the distribution $\mu_\Eps$ derived in~\eqref{eq:methods:model_error:model_error_distribution_as_pushforward_under_M}, we can show that the expectation is correct when $F$ and $f$ are affine.
    Defining $X^{\star} := \mathbb{E}[X\given Y]$, it holds that 
    \begin{equation}
    \label{eq:methods:model_error:expectation_value_of_epsilon_given_Y_affine_case}
        \mathbb{E}[\Eps\given Y] = \mathbb{E}[M(X)\given Y] = M(X^{\star})~.
    \end{equation}
    From this follows that $\mathbb{E}[Y\given X=X^{\star}] = f(X^{\star}) + M(X^{\star}) = F(X^{\star})$.
    Although most forward models are not linear, in many cases we can consider a local linearisation in the region where the posterior concentrates.
    This local linearisation is warranted if the posterior concentrates well enough, as illustrated in Figure~\ref{fig:methods:model_error:prediction_of_Y_for local_linearisation}.
    Here we use Taylor polynomials of first order for the linearisation of the forward models
    $F(x) = \mathcal{T}_{F,X^{\star}} + \mathcal{R}^1_{F,X^{\star}}$ and $f(x) = \mathcal{T}_{f,X^{\star}} + \mathcal{R}^1_{f,X^{\star}}$.
    Let $I \subseteq \mathcal{X}$ be the interval the posterior concentrates in, containing posterior samples with the probability $\mathbb{P}\left[ X \in I \given Y \right] = 1 - \delta_p$.
    If the interval is sufficiently small, we can assume that the forward models are well approximated by the linearisation, i.e.\ that 
    $\mathcal{R}^1_{f,X^{\star}}\le\delta_f$ and $\mathcal{R}^1_{F,X^{\star}}\le\delta_F$ in $I$.
    Then
    \begin{align}
        \mathbb{E}[F(X)\given Y]
        &= \mathbb{E}[\mathcal{T}^1_{F,X^{\star}}(X)\given Y] + \mathbb{E}[\mathcal{R}^1_{F,X^{\star}}(X)\given Y] \\
        &= F(X^{\star}) + \mathcal{O}(\delta_F(1-\delta_p) + \delta_p) \\
        &\approx F(X^{\star}) .
    \end{align}
    
    and, similarly, $\mathbb{E}[f(X)\given Y] \approx f(X^{\star})$.
    Consequently,
    $
        \mathbb{E}[\Eps\given Y,X]
        = \mathbb{E}[\Eps\given Y] \approx M(X^{\star})
    $
    and, by the tower property, we get for all $x\in I$ that
    \begin{align}
        \mathbb{E}[f(X) + \Eps\given X=x]
        &\approx f(x) + M(X^{\star}) \\
        &= f(x) + M(x) + \mathcal{O}(|I| + \delta_f + \delta_F) 
        \\
        &\approx F(x)~,
    \end{align}
    using $M(X^{\star}) = M(x) - (\partial_x M)(X^{\star}) \left( x - X^{\star}\right) - \mathcal{R}_{M,X^{\star}}(x)$.
    
    In particular, the error-corrected model $f(X) + \Eps$ converges, in first order, to the exact model $F(X)$ as the observation $Y$ becomes more informative.\footnote{Recall that ``more informative'' does not necessitate that more measurements are taken.}
    This means the posterior $\mu_{X|Y}$ converges even though the incorrect forward model $f$ is used.
    In this way, our model provides at least an approximation to the true posterior.
    This is illustrated numerically in Figure~\ref{fig:numerical_experiments:OHagan:more_accuracy_for_more_measurements}.
\end{remark}
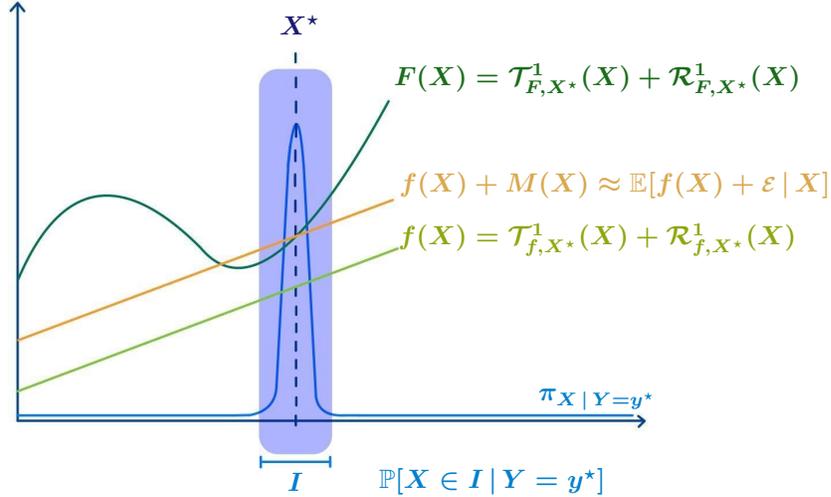
\begin{figure}[ht]
    \centering
    \begin{tikzpicture}
     \node[anchor=south west,inner sep=0] (image_local_model_error) at (0,0) {\includegraphics[clip, trim=0.5cm 40cm 0.5cm 3.5cm, width=0.75\linewidth]{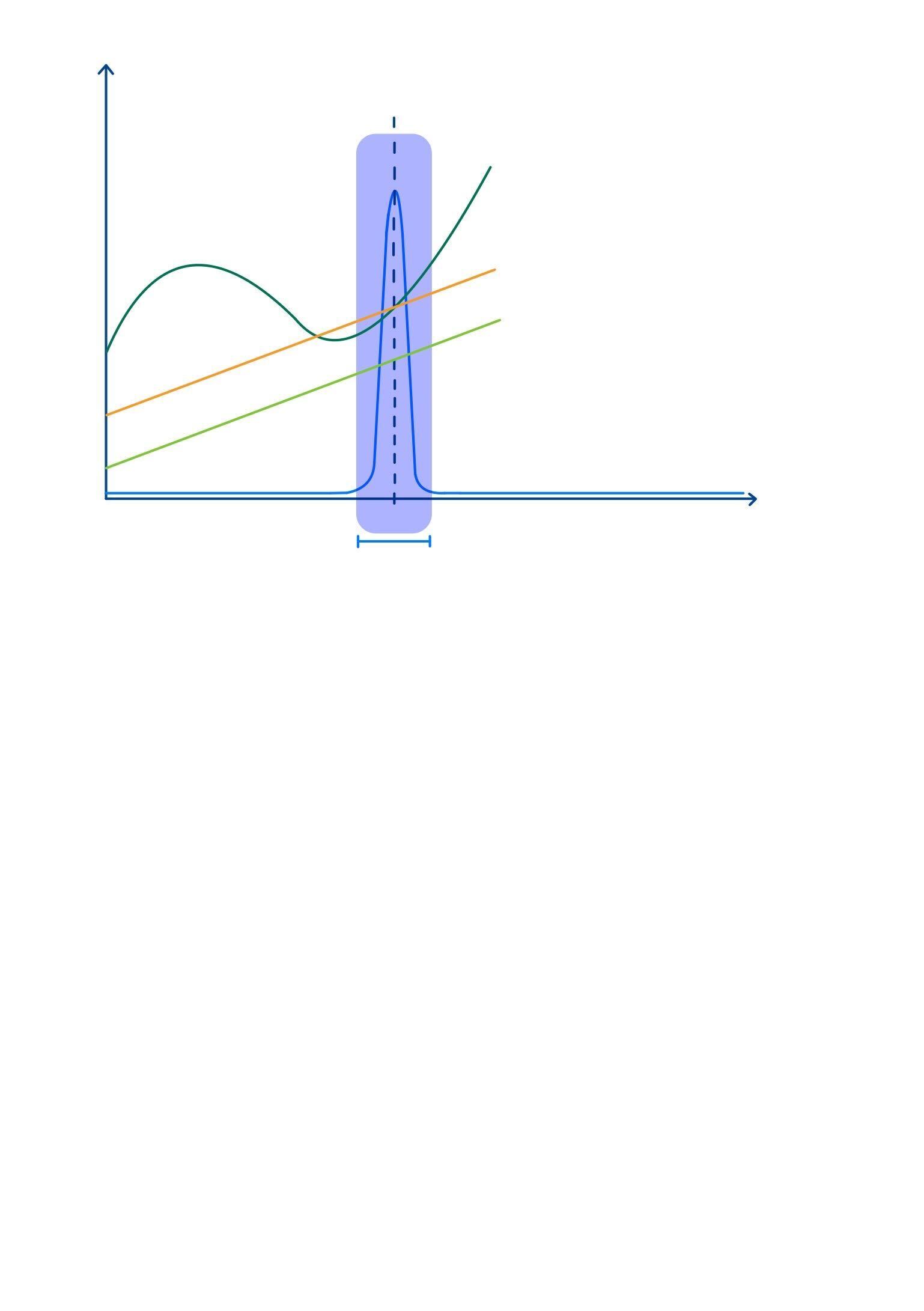}};
     \begin{scope}[x={(image_local_model_error.south east)},y={(image_local_model_error.north west)}]
            \node[text = DarkGreen, font=\bfseries] at (0.77,0.85) {\small$\boldsymbol{ F(X) = \mathcal{T}_{F,X^{\star}}^1(X) + \mathcal{R}_{F,X^{\star}}^1(X)}$};
            \node[text = QumphyOrange, font=\bfseries] at (0.79,0.66) {\small$\boldsymbol{ f(X) + M(X) \approx \mathbb{E}[ f(X) + \Eps \given X ] }$};
            \node[text = PythiaGreen, font=\bfseries] at (0.77,0.56) {\small$\boldsymbol{ f(X) = \mathcal{T}_{f,X^{\star}}^1(X) + \mathcal{R}_{f,X^{\star}}^1(X) }$};
            \node[text = NavyBlue, font=\bfseries] at (0.77,0.27) {\small $\boldsymbol{ \pi_{X \given Y = y^{\star}}}$};
            \node[text = NavyBlue, font=\bfseries] at (0.425,0.12) {\small $\boldsymbol{I}$};
            \node[text = NavyBlue, font=\bfseries] at (0.65,0.12) {\small $\boldsymbol{\mathbb{P}[X \in I \given Y = y^{\star} ]}$};
            \node[text = MidnightBlue, font=\bfseries] at (0.43,0.95) {\small $\boldsymbol{X^{\star}}$};
     \end{scope}
     \end{tikzpicture}
    \caption[Locally error-corrected model]{
    The posterior distribution $\pi_{X \given Y}$ shown here for a given observation $y^{\star}$, is concentrated in a small intervall $I$. 
    Hence the exact forward models $F$ (dark green) and the reduced one $f$ (light green) can be well approximated here by their Taylor polynomials $\mathcal{T}^1_{F,X^{\star}}$ and $\mathcal{T}^1_{f,X^{\star}}$ centered at the mean value $X^{\star}\coloneqq \mathbb{E}[ X \given Y]$ .
    The orange line indicates the reduced forward model $f(X)$ shifted by the model error prediction $M(X^{\star})$ resulting from the model error approach in~\eqref{eq:methods:model_error:model_error_distribution_as_pushforward_under_M}. 
    This provides locally a good approaximation of the exact forward model $F(X)$ in $I$, which is discussed in detail in remark~\ref{rmk:mehtods:model_error:true_prediction_for_M_linear}.
    }
    \label{fig:methods:model_error:prediction_of_Y_for local_linearisation}
\end{figure}

Using equations \eqref{eq:methods:model_error:inverse_problem_with_model_error} and \eqref{eq:methods:model_error:model_error_distribution_as_pushforward_under_M}, we can now define the Bayesian posterior including the model error.
The conditional likelihood
\begin{equation}%
    \label{eq:methods:model_error:def_conditional_likelihood}
    \pi_{Y|X,\Eps}(y\given x,\eps) = \pi_{\Eta}(y-f(x)-\eps)
\end{equation}
describes the distribution of $Y$ conditioned on $X$ and $\Eps$.
Using the change-of-variables formula~\cite[Theorem 1.6.12]{Ash1972_Realanalysisprobability} the likelihood
$\pi_{Y|X}(y\given x)$ can thus be written as the marginal
\begin{align}%
    \pi_{Y|X} (y\given x)
    &= \int_{M(\mathcal{X})} \pi_{\Eta}(y-f(x)-\eps) \dd\mu_{\eps}(\eps) \\
    &= \int_{M(\mathcal{X})} \pi_{\Eta}(y-f(x)-\eps) \dd\!\left(M_{\#} \mu_{X|Y}\right)(\eps\given y) \\
    & = \int_{\mathcal{X}} \pi_{\Eta} \big(y-f(x)-M(z) \big) \dd\mu_{X|Y}(z\given y) 
    \\
    & = \int_{\mathcal{X}} \pi_{\Eta} \big(y-f(x)-M(z) \big) \pi_{X|Y}(z\given y) \dd z ~.
\label{eq:methods:model_error:def_piM_depending_true_posterior}
\end{align}
Consequently, the resulting Bayes posterior is given by the fixed-point equation
\begin{equation}
\label{eq:methods:model_error:posterior_ModelError}
	\pi_{X|Y}(x\given y) = \frac{1}{\pi_Y(y)} \int_{\mathcal{X}} \pi_{\Eta} \big(y-f(x)-M(z) \big) \pi_{X|Y}(z\given y)\dx{z} ~\pi_X(x)~,
\end{equation}
with $\pi_{Y}(y) = \int_{\mathcal{X}} \int_{\mathcal{X}} \pi_{\Eta} \big(y-f(x)-M(z) \big) \pi_{X|Y}(z\given y) \pi_X(x) \dd z \dd x$.
Solving this fixed-point equation iteratively as in~\cite{Calvetti2018_Iterativeupdatingmodel} leads to an iterative update of the model error distribution for each update of the posterior.

\subsection{Transport maps}%
\label{subsec:methods:transport_maps}

In~\cite{Calvetti2018_Iterativeupdatingmodel} the posterior samples generated by~\eqref{eq:methods:model_error:posterior_ModelError} are generated using MCMC methods.
Although MCMC is generally very reliable, it can be very time-intensive.
To alleviate this issue, we propose a transport map-based reformulation of the iteration proposed in~\cite{Calvetti2018_Iterativeupdatingmodel}.
The subsequent section briefly explains the necessary concepts.

Given two measures, $\muref$ and $\mutar$ on $\left(\mathcal{X}, \mathcal{B}\right)$, a \emph{transport map}~\cite{Villani2009_Optimaltransportold,Marzouk2016_introductionsamplingvia,Moselhy2012_BayesianInferenceOptimal,Santambrogio2015_OptimalTransportApplied} from the \emph{reference measure} $\muref$ to the \emph{target measure} $\mu_{\text{tar}}$ is a measurable function $T \colon \mathcal{X} \to \mathcal{X}$ satisfying
\begin{equation}
\label{eq:methods:transport_maps:def_measure_transport}
    \mutar = T_{\#} \muref \coloneqq \muref \circ T^{-1}~.
\end{equation}

Although transport maps can be defined in a much more general setting (cf.~Remark~\ref{rmk:methods:transport_maps:existence_uniqueness}), we constrain ourselves to the setting where the densities $\piref$ and $\pitar$ are continuously differentiable.
Then there exists a transport $T$, namely the Knothe--Rosenblatt transport~(KR~transport) discussed in~Remark~\ref{rmk:methods:transport_maps:KR_transport}, which is a diffeomorphism and satisfies the \emph{transport of the densities} given by
\begin{equation}
\label{eq:methods:transport_maps:def_density_transport}
    \pitar = T_{\#} \piref \coloneqq \piref \circ T^{-1} \, \vert\det \grad T^{-1}\vert
\end{equation}
and similarly for the reverse transport $\muref = T^{-1}_{\#} \mutar$.
Note that, similar to~\cite{Marzouk2016_introductionsamplingvia}, we may additionally assume $T$ to be \emph{triangular}, i.e.\ to have the structure
\begin{equation}
\label{eq:methods:transport_maps:def_triangular_map}
    T(x_1, \ldots ,  x_n) = 
    \begin{pmatrix}
    T_1 (x_1) \\
    T_2 (x_1,x_2) \\
    \vdots \\
    T_n (x_1, \ldots, x_n)
    \end{pmatrix}~.
\end{equation}
This triangularity assumption is motivated by the desire to simplify the computation of $T_{\#} \piref$, as in this case
\begin{equation}
    \label{eq:methods:transport_maps:det_term_for_triangular_map}
    \det \nabla T = \prod_{i=1}^d \frac{\partial T_i}{\partial x_i}~.
\end{equation}

\begin{remark}[Knothe--Rosenblatt transport (KR transport)]
\label{rmk:methods:transport_maps:KR_transport}
    An often discussed measure transport is the Knothe--Rosenblatt transport~\cite{Knothe1957_Contributionstheoryconvex,Rosenblatt1952_RemarksMultivariateTransformation} which can be explicitly constructed and exhibits desirable properties, discussed in~\cite{Villani2009_Optimaltransportold,Santambrogio2015_OptimalTransportApplied,Zech2022_SparseApproximationTriangular,Marzouk2016_introductionsamplingvia}.
    It is monotone for lexicographic order\footnote{
    Here we stick to the definition used in~\cite{Santambrogio2015_OptimalTransportApplied}: A function $f\colon \mathbb{R}^d \to \mathbb{R}^d$ is monotone for lexicographic order if for each $j \in \{1,\ldots, d\}$ it holds, that $f_j$ is monotone in the $j\text{-th}$ component $x_j$.}
    and has the triangular structure illustrated in~\eqref{eq:methods:transport_maps:def_triangular_map}.
    Moreover, the KR transport $T$ has the same regularity as the densities $\piref$ and $\pitar$ (see~\cite[Remark~2.19]{Santambrogio2015_OptimalTransportApplied}).
    Hence, if $\piref$ and $\pitar$ are continuously differentiable, the KR-transport $T$ is a diffeomorphism.
    Even though the KR-transport can be explicitly constructed, it is very time-consuming and infeasible for higher dimensions. 
    Therefore, we rather use a numerical approximation of the KR-transport. 
\end{remark}
\begin{remark}[Existence and uniqueness of transport maps]
\label{rmk:methods:transport_maps:existence_uniqueness}
    In order to define a measure transport, at least the reference measure $\muref$ has to be atomless~\cite{Villani2009_Optimaltransportold}.
    Since we assume all measures to be absolutely continuous with respect to the Lebesgue measure, there does always exist a transport map and its reverse transport in our setting.
    To obtain uniqueness, we need to restrict the set of transport maps we consider, as for example in~\cite{Moselhy2012_BayesianInferenceOptimal,Santambrogio2015_OptimalTransportApplied}.
    In our case we choose a set, which in general does not provide a unique solution.
    As we are not looking for the optimal transport but any transport map satisfying~\eqref{eq:methods:transport_maps:def_measure_transport}, uniqueness is not necessarily required.
    As a consequence the algorithm may not return the same transport map in every execution, but as we are not primarily interested in the transport map but in the resulting distribution, we only have to ensure, that $T_{\#} \piref$ is close enough to $\pitar$.
\end{remark}

In the following section, we only consider transport maps from the model class
\begin{equation}
    \label{eq:methods:transport_maps:def_model_class_for_TM}
    \mathcal{T} \coloneqq \{T \colon \mathcal{X} \to \mathcal{X} \,\vert\, ~T \text{ triangular diffeomorphism}\}
\end{equation}
This set has beneficial analytical properties due to its smoothness and the triangular structure, which makes the computation of $\det \nabla T$ easy.
Furthermore, the KR transport is contained in $\mathcal{T}$, which ensures that at least one solution of \eqref{eq:methods:transport_maps:def_measure_transport} lies within this model class.
In order to avoid further restrictions on the model class, we dispense with the uniqueness of the transport map, as discussed in Remark~\ref{rmk:methods:transport_maps:existence_uniqueness}.

\section{Method}\label{sec:algorithm}
In contrast to the MCMC method used in~\cite{Calvetti2018_Iterativeupdatingmodel}, we propose to construct a transport map between a fixed, easy-to-sample-from density $\piref$ and the target posterior $\pitar = \pi_{X|Y}(\smallbullet\,|\, y)$ for a fixed observation $y$.
Since the explicit construction of a transport map might be very involved or even inaccessible~\cite{Carlier2008_KnothestransportBreniers,Moselhy2012_BayesianInferenceOptimal,Santambrogio2015_OptimalTransportApplied}, we aim for an approximation of the transport map by minimising
the discrepancy between $T_{\#} \piref$ and $\pitar$ as measured by the \emph{Kullback--Leibler divergence} (KL divergence)
\begin{equation}%
    \label{eq:methods:transport_maps:definition_KL-divergence}
    \Dkl\left(T_{\#}\piref \Vert \pitar \right) 
    = \mathbb{E}_{X\sim T_{\#}\piref}\Bigg[\ln\frac{T_{\#}\piref(X)}{\pitar(X)}\Bigg] 
    = \int_{\mathcal{X}} \ln \Bigl( \frac{T_{\#}\piref(x)}{\pitar(x)}\Bigr) \, T_{\#}\piref(x) \dd x.
\end{equation}
Although the KL divergence is not a metric, it attains its minimum $\Dkl(T_{\#}\piref\Vert\pitar) = 0$ if and only if $T_{\#}\piref \equiv \pitar$ $\mu_{\mathrm{ref}}$-almost everywhere~\cite{Kullback1951_InformationSufficiency}.
We therefore solve the minimisation problem
\begin{equation}
\label{eq:methods:transport_maps:optimization_problem_for_tranport_maps}
    \min_{T \in \mathcal{T}}  \Dkl(T_{\#}\piref \Vert \pitar), 
\end{equation}
with the model class $\mathcal{T}$ as in~\eqref{eq:methods:transport_maps:def_model_class_for_TM}.
\subsection{Optimisation}%
\label{subsec:algorithm:optimization}

For a general posterior density, $\pi_{X|Y}(\smallbullet \vert y)$ we want to recall the simplifications made in~\cite{Marzouk2016_introductionsamplingvia} to rewrite the optimisation~\eqref{eq:methods:transport_maps:optimization_problem_for_tranport_maps} as
\begin{align}
    &\ \min_{T \in \mathcal{T}} \Dkl \left(T_{\#} \piref \,\vert\, \pitar \right) \\
    =&\ \min_{T \in \mathcal{T}} \Dkl \left(T_{\#} \piref \,\vert\, \pi_{X|Y}(\smallbullet \given y) \right) \\
    =&\ \min_{T \in \mathcal{T}} \Dkl ( \piref \,\vert\, T^{-1}_{\#}\pi_{X|Y}(\smallbullet \given y) ) \\
    =&\ \min_{T \in \mathcal{T}} \mathbb{E}_{X \sim \piref} \Bigg[
        - \ln \pi_{Y|X}(y\given T(X))
        - \ln \pi_X(T(X)) 
        - \sum_k \ln \frac{\partial T_k}{\partial x_k}(X)
    \Bigg]
    ~.
\label{eq:algorithm:optimization:rewrite_Kl_divergence}
\end{align}

Using the KR transport $T^{\star}$ from $\piref$ to the target density $\pi_{X\given Y}(\smallbullet,y)$ and~\eqref{eq:methods:model_error:def_piM_depending_true_posterior}, the log-likelihood in~\eqref{eq:algorithm:optimization:rewrite_Kl_divergence} reads
\begin{align}
    \ln \pi_{Y|X}(y\given T(x))
    &= \ln \int_{\mathcal{X}} \pi_{\eta} \big(y-f\left(T(x)\right)-M(z)\big) \pi_{X|Y}(z\given y) \dd z \\
    &= \ln \int_{\mathcal{X}} \pi_{\eta} \big(y-f\left(T(x)\right)-M(z)\big) T^{\star}_\#\piref(z) \dd z\\
    &= \ln \int_{\mathcal{X}} \pi_{\eta} \big(y-f\left(T(x)\right)-M(T^{\star}(z))\big)\piref(z) \dd z\\
    &= \ln \mathbb{E}_{Z\sim \piref}\Bigg[
        \pi_{\eta} \big(y-f\left(T(x)\right)-M(T^{\star}(Z)) \big)
    \Bigg]~.
\label{eq:algorithm:optimization:log_likelihood_model_error_case}
\end{align}
Note that the expectation value in~\eqref{eq:algorithm:optimization:log_likelihood_model_error_case} is now taken over the reference density $\piref$, instead of the posterior. This provides a simplification of the optimisation problem~\eqref{eq:algorithm:optimization:rewrite_Kl_divergence} to
\begin{equation}
    \min_{T \in \mathcal{T}} \Dkl \left(T_{\#} \piref \,\vert\, \pitar \right)
    = \min_{T \in \mathcal{T}}  L(T, T^{\star})
\end{equation}
with the loss function
\begin{equation}
\label{eq:algorithm:optimization:loss_function_general}
\begin{aligned}
    L(T, T^{\star})
    \coloneqq
    &\mathbb{E}_{X \sim \piref} \Bigg[
        - \ln \mathbb{E}_{Z\sim \piref}\Bigg[
            \pi_{\eta} \big(y-f\left(T(X)\right)-M(T^{\star}(Z))\big)
        \Bigg]\Bigg] \\
    & + \mathbb{E}_{X \sim \piref} \Bigg[
        - \ln \pi_X(T(X))
    - \sum_k \ln \frac{\partial T_k}{\partial x_k}(X)
    \Bigg]
    ~.
\end{aligned}
\end{equation}

This leads to the fixed point equation
\begin{equation}
\label{eq:algorithm:optimization:fix_point_equation}
    T^{\star} \in \operatorname*{arg\,min}_{T\in\mathcal{T}} L(T, T^{\star}) .
\end{equation}
Similar to~\cite{Calvetti2018_Iterativeupdatingmodel}, we solve the fixed point problem~\eqref{eq:algorithm:optimization:fix_point_equation} iteratively by defining the sequence $\left( T_l\right)_{l\in\mathbb{N}} \subseteq \mathcal{T}$ via
\begin{equation}
\label{eq:algorithm:optimization:interative_definition_of_T_l}
    T_{\ell+1} \in \operatorname*{arg\,min}_{T\in\mathcal{T}} L(T, T_{\ell})~,
\end{equation}
with a suitable, e.g.\ random, initialisation $T_0$.
We terminate the iteration when
\begin{equation}
\label{eq:algorithm:optimization:termination_condition_algorithm}
    \vert L(T_{l+1},T_{l}) - L(T_{l},T_{l-1}) \vert \leq \delta, 
\end{equation}
for some fixed $\delta > 0$~.

Note that, due to the special structure of the log-likelihood term~\eqref{eq:algorithm:optimization:log_likelihood_model_error_case}, the loss~\eqref{eq:algorithm:optimization:loss_function_general} contains two nested expectation values.
In the remainder of this section, we discuss two different methods to estimate this term.

\paragraph{The Jensen-loss}
The first idea is to derive an upper bound for the negative log-likelihood term in~\eqref{eq:algorithm:optimization:loss_function_general}, which is easier to optimise numerically.
Such an upper bound to optimise the loss function is often used in machine learning algorithms.
As the problem is equivalent to seeking a lower bound of the log-likelihood, this estimate is typically referred to as the \emph{evidence lower bound (ELBO)}~\cite{Kingma2014_AutoEncodingVariational,Kingma2303_UnderstandingDiffusionObjectives}.
Using Jensen's inequality~\cite{Koenigsberger2004_Analysis1}, the negative log-likelihood can be bounded by 

\begin{align}
\label{eq:algorithm:optimization:jensen_inequality_for_log_likelihood}
    -\ln \pi_{Y|X}(y\given T(x)) 
    \le - \int_{\mathcal{X}} \ln \pi_{\eta} \Big(y- f(T(x)) - (M(T^{\star}(z)) \Big) \piref(z) \dd z~.
\end{align}
Combining this estimate with the loss~\eqref{eq:algorithm:optimization:loss_function_general} and using an empirical approximation of the integral yields the \emph{Jensen-loss}
\begin{align}
\label{eq:algorithm:optimization:loss_function_jensen}
    L_{\mathrm{J}}(T,T^{\star},(x_i)_i, (z_i)_i)
    \coloneqq 
    \frac{1}{s}\sum_{i=1}^s \Bigg[
    - \ln \pi_{\eta} \Big(y - f(T(x_i)) - (M(T^{\star}(z_i)) \Big)
    \\
    - \ln \pi_X(T(x_i)) 
    - \sum_k \ln \frac{\partial T_k}{\partial x_k}(x_i) \Bigg] ~,
\end{align}
with $(x_i,z_i)\sim\piref\otimes\piref$ for $i\in\{1,\ldots,s\}$.
This estimator converges with a rate of $\mathcal{O}(s^{-1/2})$.

Note, however, that the minimiser $T_{\mathrm{J}}$ of the (proper) upper bound $L_{\mathrm{J}}$ may be different from the minimiser of the original loss $L$. 
This is illustrated in the subsequent example, inspired by~\cite{Rainforth2018_NestingMonteCarlo}.

\begin{figure}[ht]
    \centering
        \includegraphics[width=\linewidth]{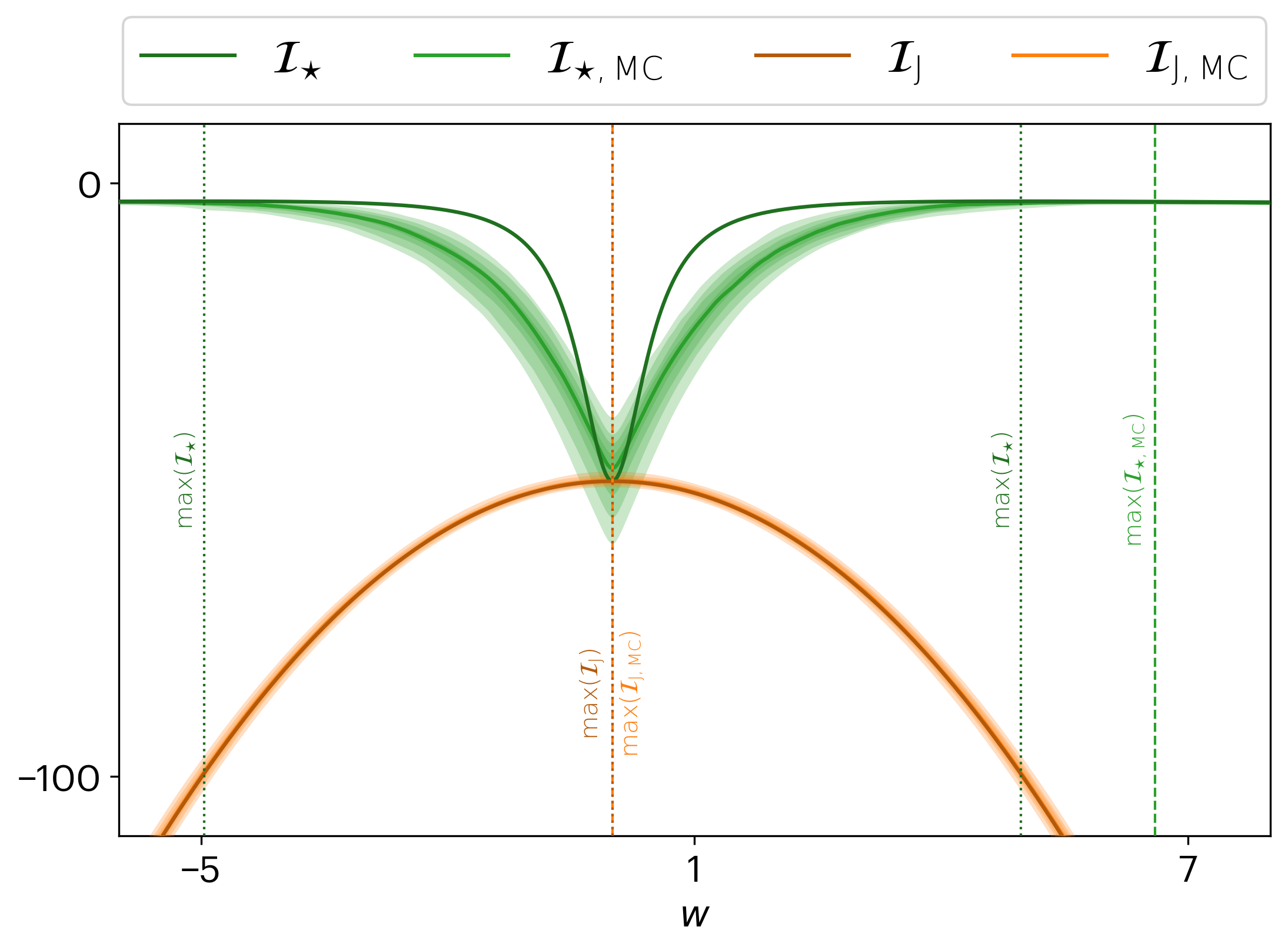}
    \caption{Graphs of the target functions $I_{\mathrm{\star}}$ and $I_{\mathrm{J}}$ of Example~\ref{ex:algorithm:optimization:bias_of_jensen_loss} with respect to $w \in \mathbb{R}$.
    $I_{\mathrm{J}}$ is estimated by the Monte Carlo estimator $I_{\mathrm{J},\mathrm{MC}}$ using $s=4^6$ i.i.d.\ samples $(x_i,y_i)_{i = 1}^s \sim \pi_X \otimes \pi_Y$.
    $I_{\mathrm{\star}}$ is estimated by the nested Monte Carlo estimator $I_{\mathrm{\star},\mathrm{MC}}$ (cf.~\cite{Rainforth2018_NestingMonteCarlo}) with i.i.d.\ samples $(x_i)_{i =1}^m \sim \pi_X$ and $(y_i)_{i=1}^n \sim \pi_Y$ where $m = \lfloor s^{\nicefrac{1}{2}} \rfloor$ and $n = \lfloor s/m \rfloor$.
    The distribution of the estimators was estimated over $100$ runs.
    }
    \label{fig:algorithm:optimization:example_jensen_vs_nmc_limit}
\end{figure}

\begin{example}[Bias of the Jensen-loss]
\label{ex:algorithm:optimization:bias_of_jensen_loss}
    Consider two random variables $Y\sim\pi_Y\coloneqq\mathcal{N}(0,1)$ and $Z\sim\pi_Z\coloneqq\mathcal{N}(0,1)$ and let the potential $\phi_w$ be given by
    \begin{equation*}
        \phi_w(y,z) = \sqrt{ \frac{2}{\pi} } \exp\left( -2 (5y+wz)^2 \right).
    \end{equation*}
    We now consider the respective maxima over $w \in \mathbb{R}$ for the two expressions
    \begin{align}
        I_{\mathrm{\star}}(w) 
        & = \mathbb{E}_y \Big[\, \ln \mathbb{E}_z \big[\, \phi_w(y,z) \,\big] \,\Big]
        = -\frac{1}{2}\ln \left( (2 w^2 + \frac{1}{2})\pi \right)- \frac{25}{2w^2+ \frac{1}{2}},
        \\
        I_{\mathrm{J}}(w) 
        & = \mathbb{E}_y \Big[\, \mathbb{E}_z \big[\, \ln \phi_w(y,z) \,\big] \, \Big]  = \mathbb{E}_{(z,y)} \left[ 2(y-z)^2 \right]
        = \frac{1}{2} \ln \left(\frac{2}{\pi}\right)-2 (25+w^2),
    \end{align}
    where $I_{\mathrm{J}}$ follows from the Jensen-inequality~\eqref{eq:algorithm:optimization:jensen_inequality_for_log_likelihood} and hence defines a lower bound with \mbox{$I_{\mathrm{\star}}\geq I_{\mathrm{J}}$}.
    Figure~\ref{fig:algorithm:optimization:example_jensen_vs_nmc_limit} illustrates that the MC-estimates of the respective expectations converge to the correct values, but that those maxima are indeed not attained in the same point.
    Although both MC estimators converge to their corresponding analytic solutions, the estimate $I_{\mathrm{J}, \mathrm{MC}}$ of $I_{\mathrm{J}}$ is unbiased and has a smaller variance than the estimate $I_{\mathrm{\star},\mathrm{MC}}$ of $I_{\mathrm{\star}}$.
    Note, however, that the underlying objective is to maximise $I_{\mathrm{\star}}$ and that the graphs of $I_{\mathrm{\star},\mathrm{MC}}$ lie much closer to $I_{\mathrm{\star}}$ than the graphs of $I_{\mathrm{J},\mathrm{MC}}$.
    The figure clearly shows the bias induced by Jensen's inequality, leading to a completely different maximiser.
\end{example}

\paragraph{The nMC-loss}
To avoid the bias introduced by Jensen's inequality, we can directly estimate both expectations in~\eqref{eq:algorithm:optimization:loss_function_general} using Monte--Carlo, yielding a nested Monte--Carlo estimator~\cite{Rainforth2018_NestingMonteCarlo}
\begin{equation}
    \ln \pi_{Y|X}\left(y\given T(x_i) \right)
    \approx \ln \frac{1}{m} \sum_{j = 1}^m  \pi_{\eta}\Big( y - f(T(x_i)) - M\big(T^{\star}(z_{ij})\big) \Big)
\end{equation}
with samples $x_i \sim \piref$ and $z_{ij} \sim \piref$ for $i \in \{1,\ldots,n\}$ and $j \in \{1,\ldots,m\}$.

Defining the potential
\begin{equation}
\label{eq:potential}
    \phi_{T,T^{\star}}(x_i, z_{ij}) \coloneqq \ln\pi_{\eta}\left( y - f(T(x_i)) - M(T^{\star}(z_{ij})) \right)\,,
\end{equation}
the resulting expression contains the logarithm of an empirical mean of an exponential function
\begin{align}
    \ln \pi_{Y|X}\left(y\given T(x_i) \right)
    &\approx \ln \frac{1}{m} \sum_{j = 1}^m  \pi_{\eta}\Big( y - f(T(x_i)) - M\big(T^{\star}(z_{ij})\big) \Big) \\
    &= \ln \sum_{j=1}^m \exp \bigg( \phi_{T,T^{\star}}( x_i, z_{ij} ) \bigg) - \ln m~,
\end{align}
To stabilise this expression, we use the standard \emph{log-sum-exp trick}
$$
    \ln \left(\sum_{j=1}^m \exp(v_j)\right)
    = \ln \left(\sum_{j=1}^m \exp (v_j - {\textstyle \max_j} v_j)\right) + {\textstyle\max_j}v_j
$$
and define the \textit{nested Monte-Carlo loss} (nMC-loss) as
\begin{equation}
\label{eq:algorithm:optimization:def_nmc_loss}
\begin{aligned}
    L_{\mathrm{nMC}}(T,T^{\star},(x_i)_i, (z_{ij})_{ij}) \coloneqq & \frac{1}{n} \sum_{i=1}^n \Bigg[~
        - \ln \sum_{j=1}^m \Big[~\exp \left( \phi_{T,T^{\star}}(x_i, z_{ij})  - {\textstyle\max_j} \phi_{T,T^{\star}}(x_i, z_{ij}) \right) \\
        &+ {\textstyle\max_j} \phi_{T,T^{\star}}(x_i, z_{ij})~\Big] 
        - \ln \pi_X(T(x_i))
    - \sum_k \ln \frac{\partial T_k}{\partial x_k}(x_i)
    \Bigg]
    - \ln m
    ~.
\end{aligned}
\end{equation}

Note that, although the term $-\ln m$ is not needed for optimisation, it is necessary for convergence of the nested Monte Carlo estimator $L_{\mathrm{nMC}}$ to the original loss $L$.
The total amount of samples required in~\eqref{eq:algorithm:optimization:def_nmc_loss} is $s=nm$.
Assuming $\phi_{T,T^{\star}}$ is thrice differentiable, the optimal convergence rate $\mathcal{O}(s^{-2/3})$ for $L_{\mathrm{nMC}}(T,T^{\star})$ is attained for $n\propto m^2$ (cf.~\cite{Rainforth2018_NestingMonteCarlo}).
Otherwise, the rate $\mathcal{O}(s^{-1/3})$ is attained for $n\propto m$ (cf.~\cite{Rainforth2018_NestingMonteCarlo}).

\subsection{Empirical iteration}
The fixed point iteration~\eqref{eq:algorithm:optimization:interative_definition_of_T_l} is solved iteratively for the analytic loss function~\eqref{eq:algorithm:optimization:loss_function_general}.
The empirical loss functions $L_{\mathrm{Jensen}}$ and $L_{\mathrm{nMC}}$ depend in addition on the drawn samples $(x_i)_i$ and $(z_{i})_{i}$ respectively $(z_{ij})_{ij}$ for the nMC-loss. 
Hence, there are several ways to define the empirical iteration depending on how the samples are drawn.
Let $\square \in \{\mathrm{J}, \mathrm{nMC}\}$ refer to the different loss functions and let $(x^{(\ell)}_i)_i \sim \piref$ and $(z^{(\ell)}_{ij})_{ij} \sim \piref$ define the used samples in each iteration, with $(z^{(\ell)}_{ij})_{ij} = (z^{(\ell)}_{i1})_{i1}$ for $\square = \mathrm{J}$.
One way to define the iteration is to use new samples in each step, i.e.\ define
\begin{equation}
\label{eq:algorithm:empirical_iteration:iteration_TL_iterative}
    T_{\ell +1} \in \operatorname*{arg\,min}_{T\in\mathcal{T}} L_{\square}(T, T_{\ell}, (x^{(\ell)}_i)_i, (z^{(\ell)}_{ij})_{ij}).
\end{equation} 

Another way is to use the same samples $(z^{(\ell)}_{ij})_{ij}$ for several iterations $\ell+k$, for $k \in \{1, \ldots K\}$, while the samples $(x^{(\ell + k)}_i)_i$ are updated in each iteration.
The corresponding iteration is then given as
\begin{equation}
\label{eq:algorithm:empirical_iteration:iteration_TL_iterative_mod}
    T_{\ell + k } \in \operatorname*{arg\,min}_{T\in\mathcal{T}} L_{\square}(T, T_{\ell}, (x^{(\ell + k)}_i)_i, (z^{(\ell)}_{ij})_{ij}),
\end{equation}
where the optimisation in step $\ell + k + 1$ is initialised by $T_{\ell +k}$.

\subsection{Computational cost}

\begin{center}
\begin{table}[ht]
    \begin{tabular}{|p{3cm}|p{7cm}|p{4.2cm}| }
        \hline
        & & \\[-5pt]
        ~                  &       \textbf{transport map approach}                &    \textbf{MCMC approach}     \\[5pt]
        \hline
        & & \\[-5pt]
        \textbf{Reduced}          &  $\mathcal{O}\big( N s G_{-} + S \big)$  &   $\mathcal{O}\big( S G_{-} \big)$ \\[5pt]
         \hline
         & & \\[-5pt]
         \textbf{Iterative }      &  $\mathcal{O}\big( N s \left( G_{-} + G_{+}  \right) + S \big)$    &   $\mathcal{O}\big( L s \left( G_{+} + G_{-} \right) + SG_-   \big)$   \\[5pt]
         \hline
         & & \\[-5pt]
         \textbf{Iterative$_{\mathbf{mod}}$}  &  $\mathcal{O}\big(L_{\mathrm{mod}} s \left( G_- + G_+ \right) + L_{\mathrm{mod}}N_{\mathrm{mod}} s G_- + S \big)$  &  --- \\[5pt]
         \hline
         & & \\[-5pt]
         \textbf{Full}              & $\mathcal{O}\big( N s G_+ + S \big)$      &   $\mathcal{O} \big( S G_+ \big)$\\[5pt]
         \hline
    \end{tabular}
    \caption{Comparison of the costs between the different methods to generate $S$ independent samples from the posterior.
    We compare the method using the transport map approach presented in this work to the one from~\cite{Calvetti2018_Iterativeupdatingmodel} using an MCMC algorithm and distinguish between the \emph{reduced} model referring to~\eqref{eq:methods:model_error:reduced_inverse_problem}, the \emph{full} model referring to~\eqref{eq:methods:bayesian_inversion:inverse_problem_exact_model} and two versions of iterative models, \emph{Iterative} given by~\eqref{eq:algorithm:optimization:interative_definition_of_T_l} and \emph{$Iterative_{\mathrm{mod}}$} defined in~\eqref{eq:algorithm:empirical_iteration:iteration_TL_iterative_mod}.
    The costs to evaluate the exact and reduced forward models $F$ and $f$ are given by $G_+$ and $G_-$, respectively.
    Furthermore, the number of function evaluations depend on the number of iterations $L_{\mathrm{mod}}$ used for the transport map approach or $L$ for the MCMC approach. 
    In each iteration, $s$ defines the number of samples used for the Monte--Carlo approximation in~\eqref{eq:algorithm:optimization:loss_function_jensen} and~\eqref{eq:algorithm:optimization:def_nmc_loss}. 
    For the transport maps approach we have to take into account the number of optimisation steps given by $N$ for the \emph{reduced}, \emph{full} and \emph{iterative} model. 
    For the algorithm \emph{Iterative} the number of optimization steps equals the number of iterations.
    For \emph{$Iterative_{\mathrm{mod}}$}, the number of iterations is given by $L_{\mathrm{mod}}$ and each iteration performs $N_{\mathrm{mod}}$ optimization steps.
    We omit the cost to generate samples from a standard normal distribution as well as the evaluation of the transport maps, since these are negligible.}
    \label{fig:algortithm:compare_cost:table_costs}
    \label{tab:algortithm:compare_cost:table_costs}
\end{table}
\end{center}

To evaluate whether the proposed iterative model actually reduces computation time, we have listed the computational cost for different models in table~\ref{fig:algortithm:compare_cost:table_costs} (compare to \cite[Table~3]{Calvetti2018_Iterativeupdatingmodel}).
The transport map approach presents a clear advantage when it is necessary to generate a large number of samples. 
Since the training of the transport map is done in an offline phase, generating posterior samples in the online phase only takes the time required to draw samples from the reference density and to evaluate the transport map.
In contrast to the MCMC approach, it is not necessary to evaluate any of the two forward models during the online phase~(see table~\ref{fig:algortithm:compare_cost:table_costs}).
This is a general advantage of transport maps and also applies to the full model (\emph{Full}) from~\eqref{eq:methods:bayesian_inversion:inverse_problem_exact_model} using the exact forward model.
Furthermore, as the transport map is updated at each optimisation step, the number of iterations for the iterative model (\emph{Iterative}), as defined in~\eqref{eq:algorithm:empirical_iteration:iteration_TL_iterative}, is equal to the number of optimisation steps $N$.
Hence, for \emph{Iterative} the exact forward model is evaluated just as often as for \emph{Full}.

To avoid this, we propose a modified iterative algorithm \emph{Iterative$_\mathrm{mod}$} as defined in~\eqref{eq:algorithm:empirical_iteration:iteration_TL_iterative_mod}
Here several optimization steps for $k \in \{1, \ldots N_{\mathrm{mod}}\}$ are done for the same model error samples $M(T_l(z^{(\ell)}_i))$ respectively $M(T_l(z^{(\ell)}_{ij}))$.
Hence within these optimization steps the exact forward model has not to be evaluated,
which reduces strongly the number of function evaluations of $F$, as shown in Table~\ref{fig:algortithm:compare_cost:table_costs}.

\section{Numerics}\label{sec:num_experiemnts}

In the subsequent section, we investigate the presented algorithm on two simple yet illustrative problems. 
For the first example, we choose both forward models $F$ and $f$ as affine linear functions, in order to verify the predictions from Remark~\ref{rmk:mehtods:model_error:true_prediction_for_M_linear}.
The second example, first presented by Brynjarsdóttir and O'Hagan~\cite{Brynjarsdottir2014_Learningphysicalparameters}, is a classical and well known benchmark example in the area of Bayesian inference under model errors.
For both examples, we compare the performance of the Jensen-loss~\eqref{eq:algorithm:optimization:loss_function_jensen} and the nMC-loss~\eqref{eq:algorithm:optimization:def_nmc_loss}.
The observational data $Y$ is generated as
\begin{equation}
    \label{eq:numerical_experiemnts:def_synthetic_data}
    y_{\mathrm{syn}} = F(x^{\star}) + \eta~, 
\end{equation}
for $\eta \sim \pi_{\Eta}$ and a fixed ``true'' parameter value $x^{\star} \in \mathcal{X}$.

For the numerical implementation, we parametrise the transport maps $T \in \mathcal{T}$ in~\eqref{eq:methods:transport_maps:def_model_class_for_TM} by invertible Neural Networks (INNs)~\cite{Kruse2021_HINTHierarchicalInvertible,Andrle2021_InvertibleNeuralNetworks}.
We want to point out that this is just one possible choice for the parametrisation and is not essential for the approach.
The chosen architecture ensures that we can analytically construct an inverse map and can easily compute $\det \nabla T$.
Although the discretised transport maps and their respective inverse are not differentiable as required in~\ref{eq:methods:transport_maps:def_model_class_for_TM}, we can still approximate all transport maps $T \in \mathcal{T}$ like the KR-transport.
For the model problems in this section, we use $L=4$ concatenations of transport and permutation layers. The used subnets are chosen to be small fully connected feed-forward neural networks with 2 hidden layers of width $10$ turned out to be sufficient for all numerical experiments.
As a consequence the training routine can be performed on a standard laptop computer within about 30 minutes.
\subsection{Affine models}%
\label{subsec:numerical_experiments:lin_toy_problem}
\label{subsec:numerical_experiments:lin_gaussian_example}

This section presents three numerical experiments based on the subsequent model problem.

\begin{example}[Affine linear forward models]
\label{ex:numerical_exaperiments:linear_forward_models}
    Consider image and parameter space to be $\mathcal{X} = \mathcal{Y} = \mathbb{R}^2$.
    We choose both forward models to be affine linear functions
    \begin{align}
        \label{eq:numerical_experiments:def_forward_models}
        F(x) \coloneqq Ax + c
        \qquad\mbox{and}\qquad
        f(x) \coloneqq x,
        \qquad\mbox{where }\ A = \begin{pmatrix} 2 & 0 \\ 0 & 3 \end{pmatrix}\mbox{and }c = \begin{pmatrix}5 \\ 5\end{pmatrix}.
    \end{align}
    The noise is assumed to be a Gaussian random variable $\Eta \sim \mathcal{N}\left(0,  \Sigma_{\Eta}\right)$ for $\Sigma_{\Eta} = (\sigma_1^2, \sigma_2^2) I_2$, with $\sigma = a F(x^{\star})$.
    The prior is chosen as $\pi_X = \mathcal{U}\left( [0,15]^2 \right)$ and the ground truth is set to $x^{\star} = \begin{pmatrix} 1,3\end{pmatrix}^T$.
\end{example}

The first example, illustrated in Figure~\ref{fig:numerical_examples:toy_problem:corner_plot_jensen_and_nmc}, compares the Bayesian posterior densities of Example~\ref{ex:numerical_exaperiments:linear_forward_models} for the full (\emph{green}) and reduced (\emph{red}) forward models as well as the posterior densities resulting from the unmodified iterative algorithm~\eqref{eq:algorithm:optimization:interative_definition_of_T_l} using the nMC-loss~\eqref{eq:algorithm:optimization:def_nmc_loss} (\emph{blue}) and the Jensen-loss~\eqref{eq:algorithm:optimization:loss_function_jensen} (\emph{orange}).
The figure illustrates that using the reduced forward model leads to a posterior far from the reconstruction value $x^{\star}$.
Both loss functions seem to be able to correct the shift of the posterior induced by the model error, but their variances are different than for the exact model.
Comparing the posterior distributions for the different loss-functions indicates that minimising the losses $L_{\mathrm{J}}$ and $L_{\mathrm{nMC}}$ leads to different results.
The more concentrated distribution associated with $L_{\mathrm{J}}$ bears the risk of being overly confident in the result.
This confidence is not warranted, as can be seen by the posterior for the full model.

\begin{figure}[ht]
    \centering
    \includegraphics[width=\linewidth]{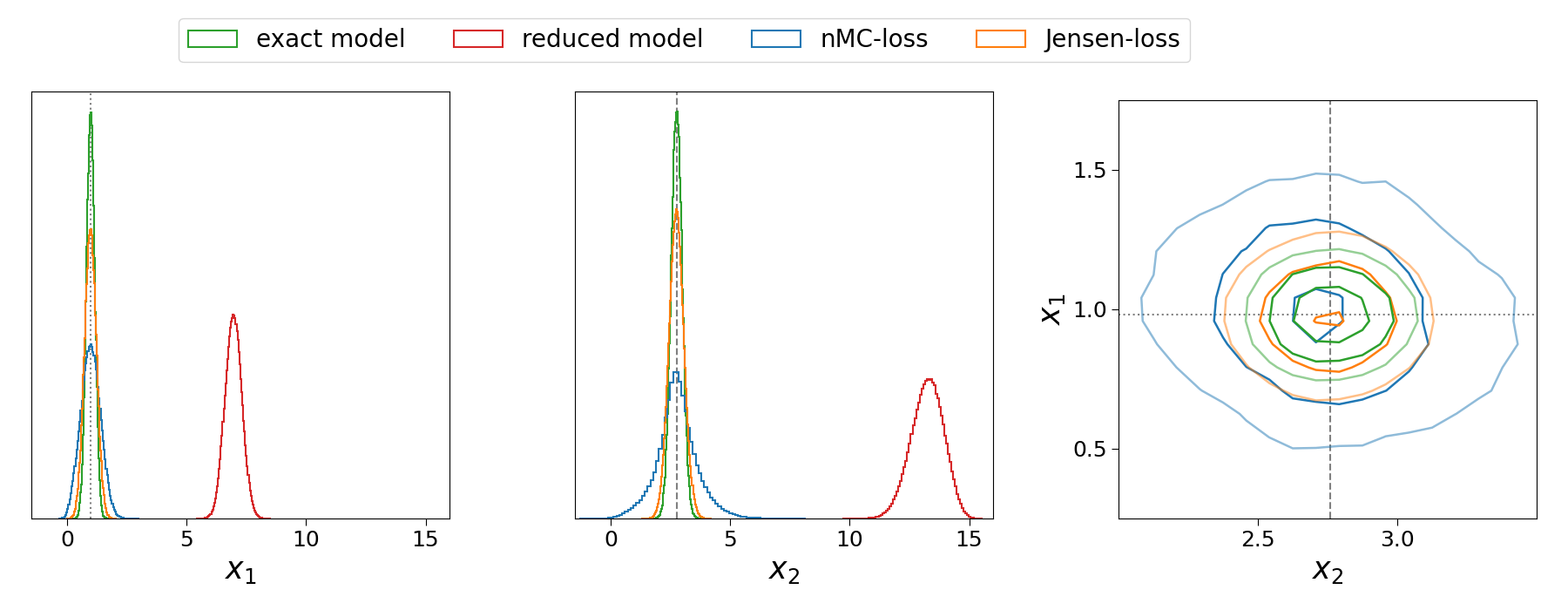}
    \caption{Posteriors for Example~\ref{ex:numerical_exaperiments:linear_forward_models}. 
    Resulting approximation of the posterior distribution of $x$ using the exact forward model and the reduced forward model, compared to the results from the updating method based on Jensen-loss $L_{\mathrm{J}}$ or nMC-loss $L_{\mathrm{nMC}}$ using $s = 10^3$ samples for the MC-approximations.}
    \label{fig:numerical_examples:toy_problem:corner_plot_jensen_and_nmc}
\end{figure}

\begin{figure}[hp]
    \centering
    \includegraphics[width=\linewidth]{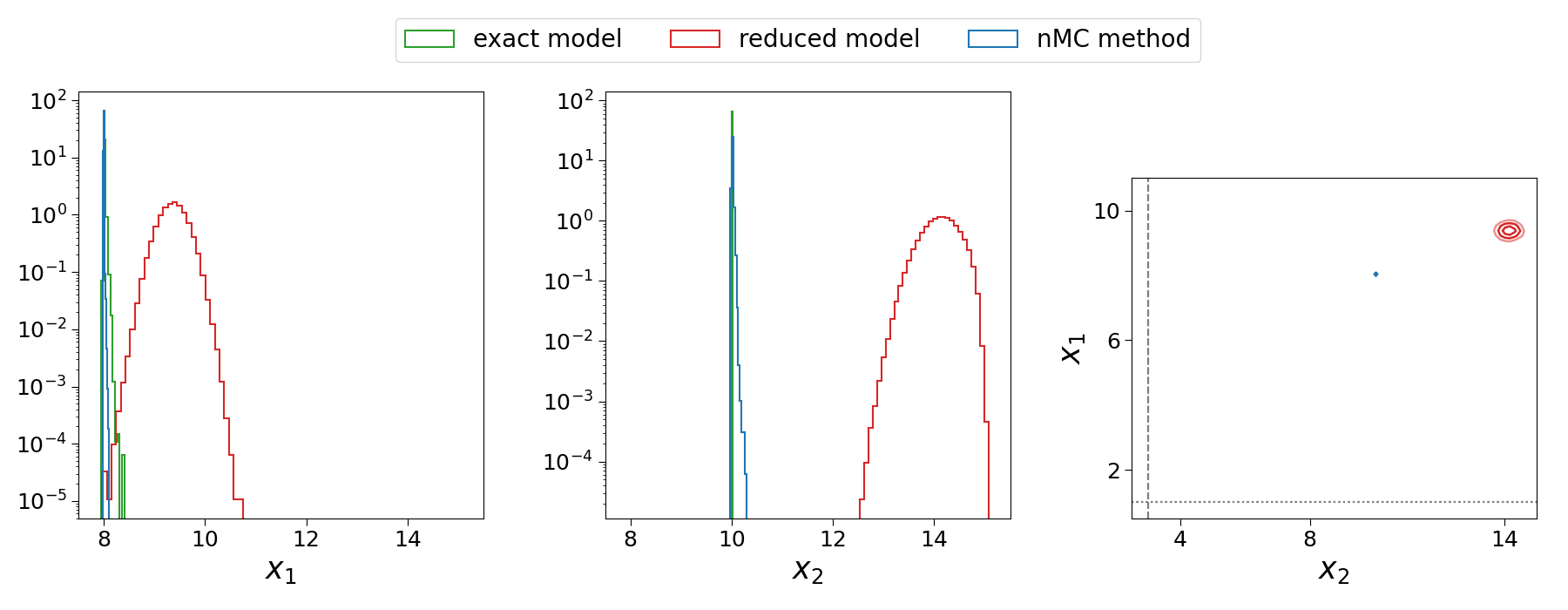}
    \caption{Posteriors for Example~\ref{ex:numerical_exaperiments:linear_forward_models} using a poorly chosen prior, given by the uniform distribution $\mathcal{U}( [8,10]\times[10,15] )$.
    By this choice, the true parameter value $x^{\star} = [1,3]^T$ is not contained in the support of the prior.
    In this setting we compare the performance of the three models: exact model (\textit{green}), reduced model (\textit{red}) and the model error approach using the nMC model (\textit{blue}).
    The one dimensional marginals are shown in logarithmic scale on the Y-axis.}
    \label{fig:numerical_experiments:linear_toy_problem:wrong_prior}
\end{figure}

For the second experiment, we choose a prior that does not contain the true value $x^{\star}$ to validate whether the iterative algorithm still improves over the reduced forward model.
Figure~\ref{fig:numerical_experiments:linear_toy_problem:wrong_prior} shows the posterior for the exact model (\emph{green}), the reduced model (\emph{red}) and the iterative algorithm using $L_{\mathrm{nMC}}$ (\emph{blue}).
We observe that the posterior from the iterative algorithm almost exactly matches the posterior for the exact model.
Both distributions are located at the boundary, the nearest possible to the true value $x^{\star}$.

\begin{figure}[htp]
    \centering
    \begin{subfigure}[b]{0.8\textwidth}
        \centering
        \includegraphics[width=\textwidth]{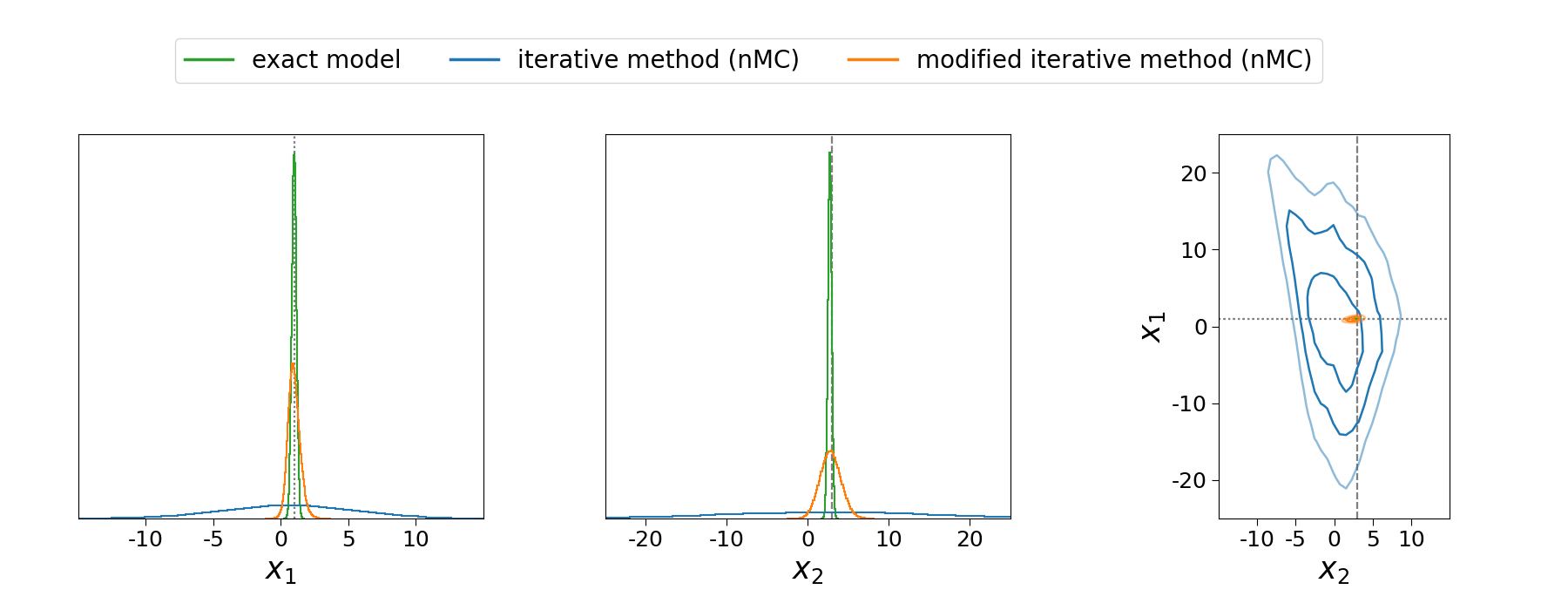}
        \caption{Comparison of the iterative models for $L_{\mathrm{mod}} = 30$ iterations for both models. }
        \label{subfig:numerics:toy_problem:iterative_models_iter_30}
    \end{subfigure}
    \hfill
    \begin{subfigure}[b]{0.8\textwidth}
        \centering
        \includegraphics[width=\textwidth]{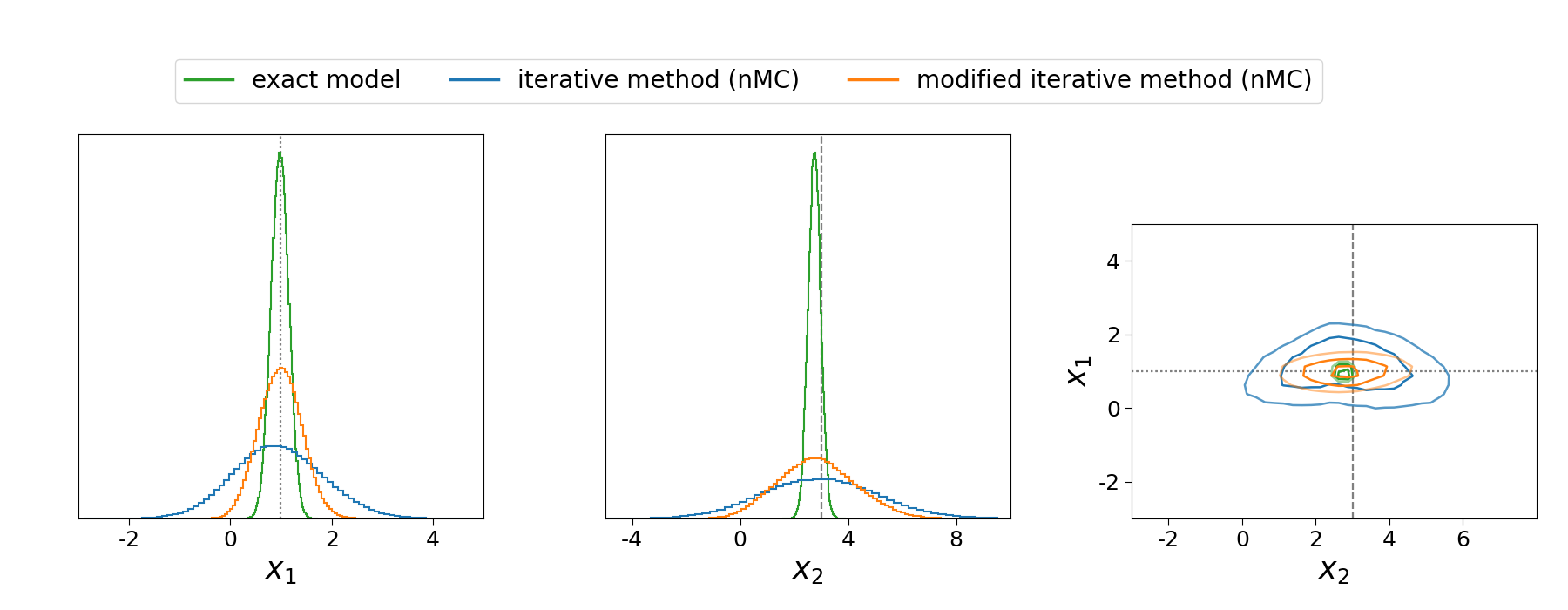}
        \caption{Comparison of the iterative models for $L_{\mathrm{mod}} = 50$ iterations for both models.}
        \label{subfig:numerics:toy_problem:iterative_models_iter_50}
    \end{subfigure}
    \hfill
    \begin{subfigure}[b]{0.8\textwidth}
        \centering
        \includegraphics[width=\textwidth]{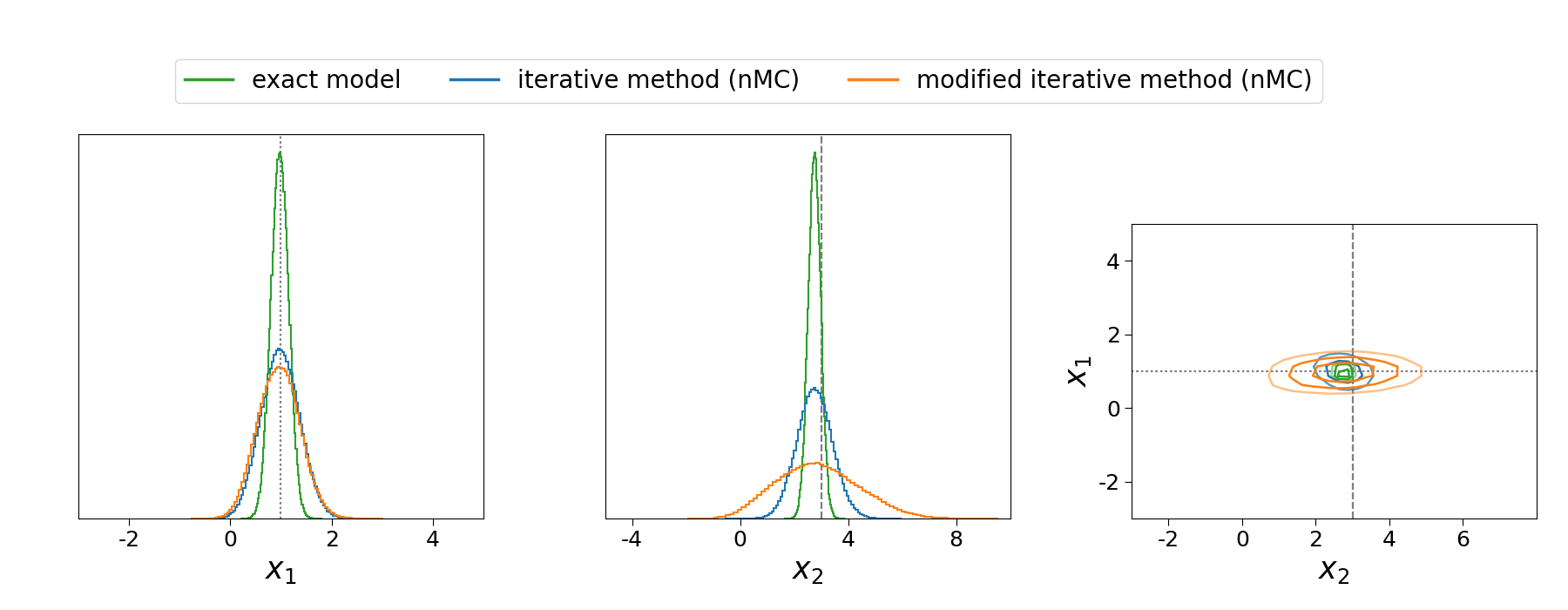}
        \caption{Comparison of the iterative models for $L_{\mathrm{mod}} = 80$ iterations for both models.}
        \label{subfig:numerics:toy_problem:iterative_models_iter_80}
    \end{subfigure}
    \caption{Posterior distributions for Example~\ref{ex:numerical_exaperiments:linear_forward_models} for three different models, the exact model (\textit{green}), the iterative model from~\eqref{eq:algorithm:optimization:def_nmc_loss} (\textit{blue})
    and the modified model from~\eqref{eq:algorithm:empirical_iteration:iteration_TL_iterative_mod} (\textit{orange}).
    Here we compare the two iterative models for three different number of iterations $L_{\mathrm{mod}}~\in~\{30, 50, 80\}$ for both iterative models.}
    \label{fig:numerics:toy_problem:compare_iterative_models_for_different_iterations}
\end{figure}

In the third numerical experiment, we compare the iterative algorithm \emph{Iterative}~\eqref{eq:algorithm:empirical_iteration:iteration_TL_iterative} 
and the modified iterative algorithm \emph{Iterative$_{\mathrm{mod}}$}~\eqref{eq:algorithm:empirical_iteration:iteration_TL_iterative_mod} to the posterior for the exact model, the \emph{exact posterior}. 
Just like the unmodified version, \emph{Iterative$_{\mathrm{mod}}$} corrects the shift induced by the model error.
Regarding the variance, \emph{Iterative$_{\mathrm{mod}}$} provides slightly elevated values, which is not a significant problem in practice.
However, the modified iteration leads to a posterior, which is not as close to the posterior for the exact model, as it is the case for the unmodified algorithm.
Only for a small number of iterations, the posterior approximation from \emph{Iterative$_{\mathrm{mod}}$} is closer to the \emph{exact posterior}.
In this sense, \emph{Iterative$_{\mathrm{mod}}$} performs better than \emph{Iterative} up to a number of $L_{\mathrm{mod}}=80$ iterations for each algorithm.
So the actual advantage of the modified algorithm is contained given a good first approximation for the posterior using as few function evaluations of the exact model as possible. 
Using this first approach as a starting point, one could use the \emph{Iterative} model to get a higher accuracy of the posterior.

\subsection{Simple machine}
\label{subsec:numerical_experiments:OHagan}

We apply the updating algorithm to the simple machine from \citeauthor{Brynjarsdottir2014_Learningphysicalparameters} presented in \cite{Brynjarsdottir2014_Learningphysicalparameters}.
Here, the forward model describes a simple machine, whose delivered amount of work is proportional to the amount of effort that is put into it.
This is, of course, an idealised concept that neglects, among other things, the friction in the machine.
Including such influences causes the exact forward model to differ from the idealised concept depending on, e.g.\ the magnitude of friction.

\begin{example}[simple machine]
    \label{ex:numerical_experiemnts:simple_machine}
    For the simple machine consider one dimensional image and parameter spaces $\mathcal{X} = \mathcal{Y} = \mathbb{R}$ and define the forward models
    \begin{align}
        F(x) \coloneqq \frac{2a\cdot x}{a+x} 
        \qquad\mbox{and}\qquad
        f(x) \coloneqq 2x~,
        \qquad\mbox{for}\ a=10.
    \end{align}
    Here, the reduced forward model $f$ is a linear approximation at $x_0 = 0$ of the exact model $F$. 
    The noise is defined by a Gaussian random variable
    \begin{equation}
        \Eta \sim \mathcal{N}\left(0, \sigma^2 \right) ~~ 
        \mathrm{with} ~~\sigma = a F(x^{\star}),
    \end{equation}
    and the prior is set to $\pi_X := \mathcal{U}\left( [0,15] \right)$.
    To investigate the performance of our algorithm regarding different model error magnitudes, we consider two different ground truth values \mbox{$x^{\star} = 3$} and $\hat{x}^{\star} = 10$.
\end{example}

For the present example we first compare the posterior approximations resulting from the two different losses $L_{\mathrm{J}}$ and $L_{\mathrm{nMC}}$ (see Figure~\ref{fig:numerical_experiemnts:ohagan_example_forward_models_and_posterios_plots}). 
Then we observe the behaviour of the resulting posterior distributions for an increasing number of measurements, shown in Figure~\ref{fig:numerical_experiments:OHagan:more_accuracy_for_more_measurements}.
The difference between the forward models increases for larger $x$ as shown in Figure~\ref{fig:numerical_experiments:OHagan:forward_models} on the prior-domain $[0,15]$.
This is also reflected in the posterior distributions for the different setups shown in Figure~\ref{fig:numerical_experiments:OHagan:posterior_all_models}. 
Comparing the two posteriors for the reduced forward model we see, that a larger model error results in a bigger shift of the posterior distribution.
Also, the iterative algorithm approximates the mean value well, independent of the model error magnitude.
Similar to Example~\ref{ex:numerical_exaperiments:linear_forward_models} the variance is not captured well by the iterative algorithm.
Moreover we see that both loss functions yield a similar approximation result. 
The main difference can be seen in the variance, which is slightly larger for the nMC-loss. 
In the second setup the mean value of the posteriors does not agree well with the true value $\hat{x}^{\star}$, even for the exact model. 
This is caused by the magnitude of the noise added to the synthetic data. 
Since the variance is underestimated by the iterative algorithm, $\hat{x}^{\star}$ is not even contained within a $95\%$ confidence interval of the posterior distributions for both $L_{\mathrm{J}}$ and $L_{\mathrm{nMC}}$.
This emphasises why an underestimation of the variance is in general problematic.

In the second numerical experiment for Example~\ref{ex:numerical_experiemnts:simple_machine} we investigate if the model error accuracy for the expectation increases if the posterior concentrates, as claimed in Remark~\ref{rmk:mehtods:model_error:true_prediction_for_M_linear}.
To get a posterior concentrated in a smaller region, we increase the number of measurement points to $y = ( y_1, \ldots, y_d )^T$ with $d = 500$.
The posterior approximation results for this setting for the example of a simple machine with $\hat{x}^{\star} = 10$ are displayed in Figure~\ref{fig:numerical_experiments:OHagan:more_accuracy_for_more_measurements}.
Here we observe that, independent of the used forward model, the variance becomes smaller for a larger number of measurements.
For the reduced model the expectation converges to the wrong value. 
This shows, once more, that using a wrong forward model induces an error which can not be compensated for by increasing the number of measurement.
On the contrary, for the exact forward model and the model error approach from~\eqref{eq:algorithm:optimization:def_nmc_loss} the mean indeed tends to the correct value $\hat{x}^{\star} = 10$. 
This confirms, that a smaller posterior actually results in greater mean value approximation accuracy for the model error approach.
In contrast to the mean values, however, the model error approach does not improve variance approximation if we increase the number of measurements. This behaviour can be observed in the right image of Figure~\ref{fig:numerical_experiments:OHagan:more_accuracy_for_more_measurements}. 
As an explanation, note that the variance of the posterior depends on the slope of the used forward model. 
Remark~\ref{rmk:mehtods:model_error:true_prediction_for_M_linear} only shows that the mean value of the posterior tends to the ground truth if the true posterior become a peak.
But since the linearization does not approximate the slope of the true forward model (cf.\ Figure~\ref{fig:methods:model_error:prediction_of_Y_for local_linearisation}) the variance does not match the true variance even though more measurement points are used.
As discussed in Figure~\ref{fig:numerical_experiments:OHagan:posterior_all_models}, an underestimation of the variance leads to an error in the reconstruction of $x$, which is an undesirable outcome.

\begin{figure}[ht]
    \centering
    \begin{subfigure}{0.3\textwidth}
        \centering
        \includegraphics[width=\linewidth]{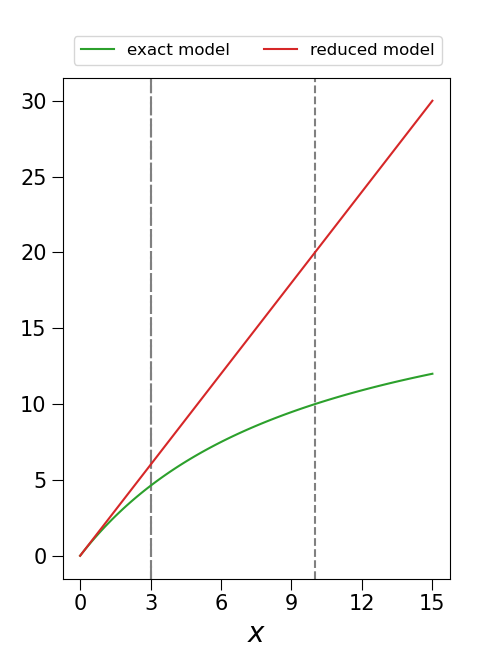} 
        \caption{Exact (\textit{green line}) and reduced (\textit{red line}) forward models $F,f$.}
        \label{fig:numerical_experiments:OHagan:forward_models}
    \end{subfigure}
    \hfill
    \begin{subfigure}{0.65\textwidth}
        \centering
        \includegraphics[width=\linewidth]{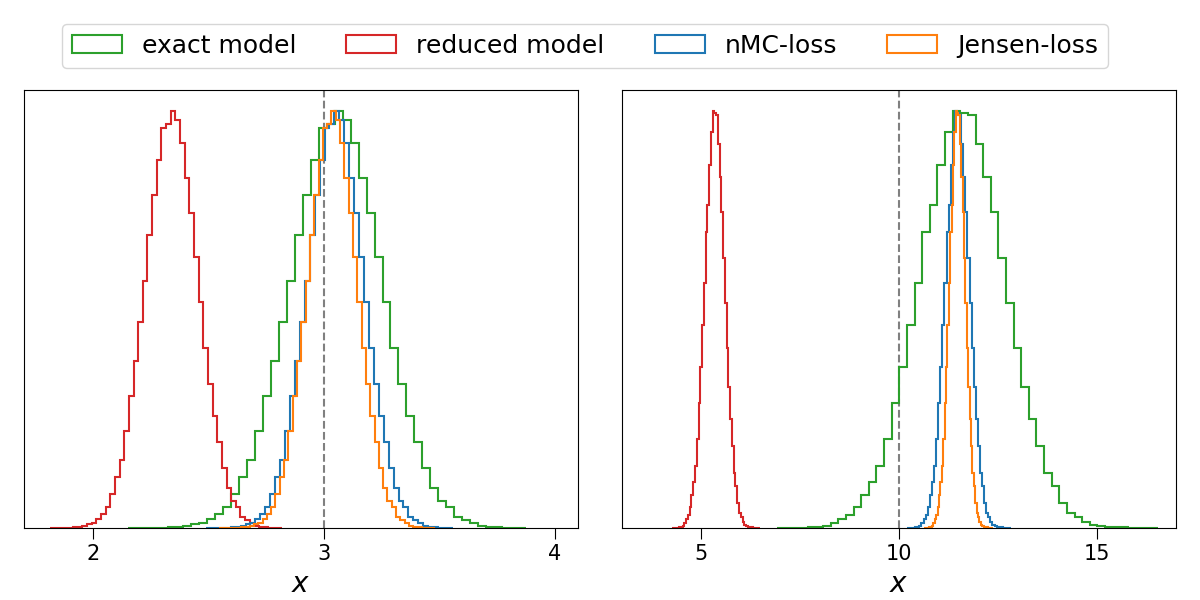}
        \caption{Posterior distributions for different models: for the exact forward model (\textit{green}),
        the reduced forward model (\textit{red}) and the updating method using either the Jensen-loss (\textit{orange}) or the nested MC-loss (\textit{blue}).
        The left image shows the posterior for the true parameter $x^{\star} = 3$ and the right one the posteriors for $x^{\star} = 10$.}
        \label{fig:numerical_experiments:OHagan:posterior_all_models}
    \end{subfigure}
    \caption{Forward models and posterior distributions for Example~\ref{ex:numerical_experiemnts:simple_machine} of the \emph{simple machine}.}
    \label{fig:numerical_experiemnts:ohagan_example_forward_models_and_posterios_plots}
\end{figure}

\begin{figure}[ht]
    \centering
    \includegraphics[width=0.6\linewidth]{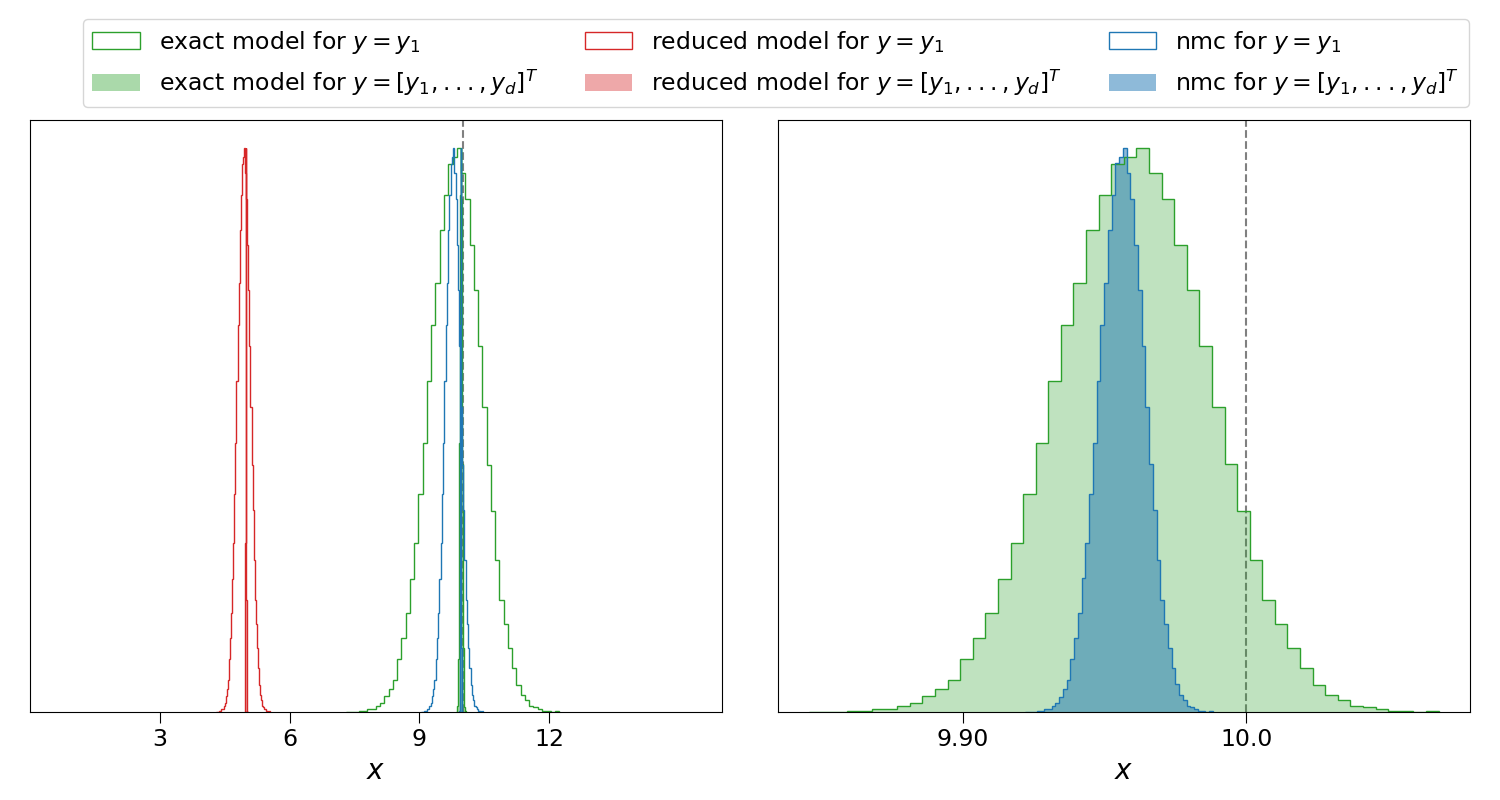}
    \caption{Posterior distributions for increasing the number of measurements for the example of the simple machine in~\ref{subsec:numerical_experiments:OHagan}.
    Here we compare the three different models: reduced model (\textit{red}), full model (\textit{green}) and nMC--model (\textit{blue}) from~\eqref{eq:algorithm:optimization:def_nmc_loss}.
    The shaded distributions are for one measurement $y =y_1 \in \mathbb{R}$ are given by the distributions, that are not filled. 
    The filled ones are resulting from the reconstruction using  $d = 500$ measurements $y = ( y_1, \ldots, y_d )^T$.
    In the right image we have enlarged the region around the true value $x^{\star} = 10$.
    So the distributions for the exact model and the nMC--model for $d=500$ measurements can be compared better.} 
    \label{fig:numerical_experiments:OHagan:more_accuracy_for_more_measurements}
\end{figure}

\section{Conclusion}\label{sec:conclusion}

In this paper we extend the model error approach presented in~\cite{Calvetti2018_Iterativeupdatingmodel} and contextualize the required assumption in relation to other approaches (see Table~\ref{tab:methods:model_error:table_models_overview}). 
As a justification that the used model error distribution provides a reasonable choice, we show in Theorem~\ref{thm:methods:model_error:model_error_distribution_follow_from_assumptions} that the model error distribution satisfies two key properties of the exact model error distribution. 

As a next step, we adapt the iterative approach from~\cite{Calvetti2018_Iterativeupdatingmodel} to use transport maps, instead of MCMC, for sampling from the posterior distribution. 
This has the advantage that we can draw many samples very fast, after training the transport map. 
The computational cost of both algorithms are compared in Table~\ref{tab:algortithm:compare_cost:table_costs}.

To construct the transport map we define two different loss functions; the Jensen-loss~\eqref{eq:algorithm:optimization:loss_function_jensen} and the nMC-loss~\eqref{eq:algorithm:optimization:def_nmc_loss} 
and compare them based on the numerical experiments in section~\ref{sec:num_experiemnts}.
We observe that both loss functions lead to similar results, even though the Jensen-loss is an empirical upper bound for our objective.
Indeed, we demonstrate in Example~\ref{ex:algorithm:optimization:bias_of_jensen_loss}, that the Jensen-loss can lead to relevant deviations from the results of the nMC-loss.
Since the convergence rate of the nMC-loss is only marginally worse than for the Jensen-loss, the nMC-loss is in general preferable.

The numerical experiments in section~\ref{sec:num_experiemnts} show that our algorithm produces posterior densities which are very close to the posterior density of the exact model.
Although the mode of the iterative posterior density fits the mode of the exact posterior very well, the variances are not correct.\philipp{What happens if the distribution is multi-modal?}
This is problematic, as it may lead to an underestimation of the variance, see Figure~\ref{fig:numerical_experiments:OHagan:posterior_all_models}.
To what extent this is a cause of the chosen forward model or of the algorithm itself is still an open question.

A downside of the iterative algorithm is the high number of exact model evaluations required. 
This increases the computational cost, as we assume that the evaluation of the exact model $F$ is computationally expensive. 
We address this problem with the modified algorithm~\eqref{eq:algorithm:empirical_iteration:iteration_TL_iterative_mod}, which requires fewer evaluations of $F$ but also looses accuracy in the resulting posterior approximation. 
To combine the best of both worlds, the modified algorithm can be used to obtain a first approximation of the posterior.
Subsequently, we can use the corresponding transport map as an initial value for a few iterations of the full iterative algorithm.
This procedure shows promise for future applications.

Another approach to address the problem of drawing posterior samples with a computationally expensive forward model is the \emph{multi-level delayed acceptence} (MLDA) method presented in~\cite{Lykkegaard2023_MultilevelDelayedAcceptance}.
Here, several forward models with different accuracy levels are used to obtain a rough posterior approximation on the coarse level and improve that approximation using models on finer levels.
This strongly reduces the number of evaluations of the expensive forward model and leads to an accurate posterior distribution. 
In contrast to our approach, the MLDA uses MCMC to sample from the posterior instead of transport maps.

The advantages of the presented algorithm are particularly evident when many posterior samples are required.
For a number of $S$ independent samples, a standard MCMC algorithm requires $\mathcal{O}(\tau S)$ forward evaluations, dependent on the autocorrelation time $\tau$. 
In contrast, our algorithm is divided in an offline and an online phase, and so the training time is independent of $S$.
For an expensive forward model, e.g.\ those provided by a finite element solver, the cost to sample from the posterior is mainly dominated by the number of forward evaluations of the exact forward model.
To train e.g.\ the transport map in Example~\ref{ex:numerical_exaperiments:linear_forward_models} adequately using \textit{Iterative}$_{\mathrm{mod}}$ adequately, it takes around $N_f = 5\cdot10^4$ forward evaluations, which results in a computation time of approximately $N_f G_+$.
Depending on the problem, larger sample sets are often required to describe the posterior sufficiently good, which increases the computational cost for the MCMC algorithm often by several orders of magnitude. %
The presented algorithm instead only needs the training time, which always stays the same, and a few seconds to generate posterior samples.
This makes it possible to draw an arbitrary large number of samples or generate more samples with nearly no cost in the online phase. 
We note that the iterative MCMC approach from~\cite{Calvetti2018_Iterativeupdatingmodel} already decoupled the number of samples $S$ from the exact forward model, but has no separation between an online and offline phase, which causes the same computational costs for each set of samples. 
In summary, the presented algorithm offers significant advantages in scenarios where an expensive forward model is involved and a large number of posterior samples are required to describe the parameter distribution.
Moreover, we emphasize that this approach would also be beneficial in situations where sampling from similar posteriors is necessary, such as in production control, as learning a transport to a similar distribution comes with almost no additional computational effort.

\appendix

\clearpage
\printbibliography

\end{document}